\documentclass[conference]{IEEEtran}
\usepackage[utf8x]{inputenc}
\usepackage{t1enc}
\usepackage{hyperref}
\hypersetup{
  colorlinks   = true,    
  urlcolor     = brown!50!black,    
  linkcolor    = brown!50!black,   
  citecolor    = red      
}

\usepackage[cm]{fullpage}
\usepackage{hhline}

\usepackage{xcolor}

\usepackage{amsfonts}
\usepackage{amssymb}
\usepackage{tikz}
\tikzset{
  treenode/.style = {shape=rectangle, rounded corners,
                     draw, align=center,
                     top color=white, bottom color=blue!20},
  root/.style     = {treenode, font=\Large, bottom color=red!30},
  env/.style      = {treenode, font=\ttfamily\normalsize},
  whitecircle/.style = {circle,draw},
  blackcircle/.style = {circle,draw,fill=black}
}

\usepackage{footnote}
\makesavenoteenv{table}
\makesavenoteenv{tabular}

\mathchardef\mhyphen="2D

\newcommand{\CC}{\ensuremath{\hyperlink{dictcc}{\textsf{CC}}}}
\newcommand{\CBI}{\ensuremath{\hyperlink{dictcc}{\textsf{WBI}}}}
\newcommand{\DC}{\ensuremath{\hyperlink{dict}{\textsf{DC}}}}
\newcommand{\DCSERIAL}[1]{\ensuremath{\hyperlink{dictdc}{\textsf{DC}}_{#1}^{\mathit{serial}}}}
\newcommand{\BI}{\ensuremath{\hyperlink{dict}{\textsf{BI}}}}
\newcommand{\BILEAST}[1]{\ensuremath{\hyperlink{dictdc}{\textsf{BI}}_{#1}^{\mathit{least}}}}
\newcommand{\GBI}{\ensuremath{\hyperlink{dictgen}{\textsf{GBI}}}}
\newcommand{\GDC}{\ensuremath{\hyperlink{dictgen}{\textsf{GDC}}}}
\newcommand{\AC}{\ensuremath{\hyperlink{dict2}{\textsf{AC}}}}
\newcommand{\coAC}{\ensuremath{\hyperlink{dict2}{\textsf{co-AC}}}}
\newcommand{\KL}{\ensuremath{\hyperlink{dict}{\textsf{KL}}}}
\newcommand{\FT}{\ensuremath{\hyperlink{dict}{\textsf{FT}}}}
\newcommand{\WKL}{\ensuremath{\hyperlink{dict}{\textsf{WKL}}}}
\newcommand{\RCA}{\ensuremath{\textsf{RCA}}}

\newcommand{\BPI}{\ensuremath{\hyperlink{axalg}{\textsf{BPI}}}}
\newcommand{\coBPI}{\ensuremath{\hyperlink{axalg}{\textsf{co-BPI}}}}
\newcommand{\BPF}{\ensuremath{\hyperlink{axalg}{\textsf{BPF}}}}
\newcommand{\coBPF}{\ensuremath{\hyperlink{axalg}{\textsf{co-BPF}}}}
\newcommand{\Compl}{\ensuremath{\hyperlink{dictcompl}{\textsf{Compl}}}}

\newcommand{\free}[1]{\mathsf{Free}(#1)}

\newcommand{\emptyseq}{\langle\rangle}
\newcommand{\mytensor}{\ensuremath{\mathbin{\land}}}
\newcommand{\mywith}{\ensuremath{\mathbin{\land}}}
\newcommand{\mymultimap}{\ensuremath{\mathbin{\imp}}}
\newcommand{\myoplus}{\ensuremath{\mathbin{\lor}}}
\newcommand{\Ax}{\mathsf{Ax}}
\newcommand{\Cut}{\mathsf{Cut}}
\newcommand{\intersect}{\ensuremath{\mathbin{)\hspace{-1.4mm}(}}}
\newcommand{\Prop}{\ensuremath{\mathsf{Prop}}}
\newcommand{\N}{\ensuremath{\mathbb{N}}}
\newcommand{\Bool}{\ensuremath{\mathbb{B}}}
\newcommand{\myifbool}[3]{\mathsf{if}\;{#1}\;\mathsf{then}\;{#2}\;\mathsf{else}\;{#3}}
\newcommand{\innerpos}[1]{\,\downarrow^{\!\scriptscriptstyle +}\!{#1}}
\newcommand{\outerofpos}[1]{\,\uparrow^{\!\scriptscriptstyle +}\!{#1}}
\newcommand{\innerneg}[1]{\,\downarrow^{\!\scriptscriptstyle -}\!{#1}}
\newcommand{\outerofneg}[1]{\,\uparrow^{\!\scriptscriptstyle -}\!{#1}}

\newcommand{\dom}[1]{\mathit{dom}(#1)}
\newcommand{\myeq}{\mathbin{\mbox{$=\!\!=$}}}
\newcommand{\Leaf}{\ensuremath{\mathtt{Leaf}}}
\newcommand{\Node}{\ensuremath{\mathtt{Node}}}
\newcommand{\nth}[1]{{#1}^{\mbox{\scriptsize th}}}

\IEEEoverridecommandlockouts
\IEEEpubid{\makebox[\columnwidth] 
  {978-1-6654-4895-6/21/\$31.00~\copyright2021 IEEE \hfill}
  \hspace{\columnsep}\makebox[\columnwidth]{ }}
 
\title{On the logical structure of choice and bar induction principles\\
\raisebox{.5cm}{\small (includes errata -- January 2026)}\vspace{-0.9cm}}

\author{
\IEEEauthorblockN{Nuria Brede}
\IEEEauthorblockA{University of Potsdam, Germany}
\and
\IEEEauthorblockN{Hugo Herbelin}
\IEEEauthorblockA{Inria Paris, Universit\'e de Paris, CNRS, IRIF, France}}

\def\proofskip{\vskip 5pt}
\def\proofbox{\hfill\rule{6pt}{6pt}}
\newcommand{\omitnow}[1]{}
\newcommand\postsoumission[1]{} 

\newenvironment{proof}{{\noindent \sc Proof:\/}~}{\proofbox\proofskip}

\newtheorem{prop}{Proposition}

\newtheorem{theorem}{Theorem}

\def\IMP{\rightarrow}
\def\arrow{\IMP}

\newcommand \seqr[3]
  {\shortstack{$#2$ \\ \mbox{}\\
                   \mbox{}\hrulefill\mbox{}\\ \mbox{}\\ $#3$} \; \raisebox{3ex}{$\;\;\mbox{$#1$}$}}
\newcommand \seq[2]{\shortstack{$#1$ \\ \mbox{}\\
                    \mbox{}\hrulefill\mbox{}\\ \mbox{}\\ $#2$}}

\newcommand\defeq{\triangleq}

\newcommand{\imp}{\Rightarrow}

\newcommand{\coupesep}{|\!|}
\newcount\coupelevel
\newcount\coupecolor
\newcommand{\coupe}[2]{\black{
  \ifnum \coupelevel>0 
  \coupecolor=\coupelevel
  \divide\coupecolor by 3\relax
  \output{\coupecolor}
    {\mbox{\textlangle}} \overline{\gray{#1} \gray{\coupesep} \gray{#2}} \mbox{\textrangle}
  \else
    \advance\coupelevel by 1\relax\langle \redcolor{#1} \coupesep \blue{#2} \rangle
    \output{\coupelevel}
  \fi}}

\begin{document}

\maketitle

\begin{abstract}

We develop an approach to choice principles and their
contra-positive bar-induction principles as extensionality schemes
connecting an 
``intensional'' or ``effective'' view of respectively
ill- and well-foundedness properties to an ``extensional'' or ``ideal''
view of these properties. After classifying and analysing the
relations between different intensional definitions of ill-foundedness
and well-foundedness, we introduce, for a domain $A$, a codomain $B$
and a ``filter'' $T$ on finite approximations of functions
from $A$ to $ B$, a generalised form $\GDC_{ABT}$ of the axiom of dependent choice
and dually a generalised bar induction principle $\GBI_{ABT}$ such
that:

$\GDC_{ABT}$ intuitionistically captures the strength of
\begin{itemize}
\item the
  general axiom of choice expressed as $\forall a\, \exists b\, R(a,b)
  \imp \exists \alpha\,\forall a\,R(a,\alpha(a)))$ when $T$ is a
  filter that derives
  point-wise from a relation $R$ on $A \times B$ without introducing
  further constraints,
\item the
  Boolean Prime Filter Theorem / Ultrafilter Theorem if $B$ is the
  two-element set $\Bool$ (for a constructive definition of prime filter),
\item the
  axiom of dependent choice if $A=\N$,
\item Weak Kőnig's
  Lemma if $A=\N$ and $B=\Bool$ (up to weak classical reasoning).
\end{itemize}  

  $\GBI_{ABT}$ intuitionistically captures the strength of
\begin{itemize}
\item
  G\"odel's completeness theorem in the form validity implies
  provability for entailment relations if $B=\Bool$ (for a
  constructive definition of validity),
\item bar induction if $A=\N$ and $T$ is decidable,
\item the Weak Fan Theorem if
$A=\N$ and $B=\Bool$.
\end{itemize}

Contrastingly, even though $\GDC_{ABT}$ and $\GBI_{ABT}$ smoothly
capture several variants of choice and bar induction, some
instances are inconsistent, e.g. when $A$ is $\Bool^{\N}$ and $B$ is
$\N$.

\end{abstract}


\section{Introduction}
\label{sec:intro}

\subsection{Bar induction, dependent choice and their 
variants as extensionality principles}

For a domain $A$, there are different ways to define a
well-founded tree branching over $A$. A first possibility is to define
it as an inductive object built from leaves  
and from nodes associating a subtree to each element in $A$. We will
call this definition {\em intensional}. Using a
syntax familiar to functional programming languages or Martin-L\"of-style type
theory, such intensional trees correspond to inhabitants of an inductive type:

\medskip
\begin{tt}
type tree = $\Leaf$ | $\Node$ of (A $\rightarrow$ tree)
\end{tt}
\medskip

A second possibility is a definition which we shall
call {\em extensional} and which is probably more standard in the
context of non type-theoretic mathematics. Let $A^*$ denote the set of
finite sequences of elements of $A$, with 
$\emptyseq$ denoting the empty sequence and $u \star a$ the extension
of the sequence $u$ with $a$ from~$A$. Then an extensional tree $T$ is
a downwards-closed predicate over~$A^*$. Finite sequences $u \in A^*$
are interpreted as finite paths from the root of a tree and the
predicate determines which paths are contained in $T$.
We say that $T$ is {\em extensionally
  well-founded} if for all infinite paths $\alpha$ in $A^{\N}$, the
path eventually ``leaves'' the tree, i.e. there is an initial finite
prefix $u$ of $\alpha$ such that $u$ (as path from the root) is not
contained in $T$.

The intensional definition is stronger: to any inductively-defined
tree $t$, we can associate an extensionally well-founded tree $T(t)$ by
recursion on $t$ as follows:
$$
\begin{array}{lll}
u \in T(\Leaf) & \defeq & \bot\\
u \in T(\Node(f)) & \defeq & \lor \!\!\begin{array}{l}u = \emptyseq
 \\ \exists a\,\exists u'\,(u = a@u' \land u' \in T(f(a)))
\end{array}\!\!\!\\
\end{array}
$$ where $a@u$, a particular case of concatenation $u'@u$, prefixes $u$ with $a$.
We can then prove by induction on $t$ that
$\forall \alpha\, \exists n\, \neg T(t)(\alpha_{|n})$, where
$\alpha_{|n}$ is the restriction of $\alpha$ to its first $n$ values.

\newcommand{\mysf}[1]{\it #1}
To reflect that $T(t)$ is related to $t$, we can define a
realisability relation between $t$ and $T$ as follows:
\begin{itemize}
\item $\Leaf$ {\mysf realises} $T$ if $\emptyseq \notin T$
\item $\Node(f)$ {\mysf realises} $T$ if $\emptyseq\in T$ and for all $a$, $f(a)$ 
    realises $\lambda u.(a @ u \in T)$
\end{itemize}
Then, we can prove by induction on $t$ that $t$
realises $T(t)$.

{\em Bar induction}, introduced by Brouwer and further analysed e.g. by Kleene and
Vesley~\cite{KleeneVesley65} can be seen as the converse property, namely that any
extensionally well-founded $T$ can be turned into an
inductively-defined tree $t$ that realises $T$, so that, at the end, the intensional
and extensional definitions of well-foundedness are
equivalent\footnote{Kleene and Vesley~\cite{KleeneVesley65} used
  respectively the terms ``inductive'' and ``explicit'' for what we
  call intensional and extensional.}.

At its core, bar induction is the statement ``$U$ barred implies $U$
inductively barred'' for $U$ a predicate on $A^*$.  As studied e.g. in
Howard and Kreisel~\cite{HowardKreisel66}, when used on a negated
predicate $\neg T$, this reduces to ``$T$ extensionally well-founded
implies $T$ inductively well-founded'', where $T$ {\it inductively
  well-founded} abbreviates ``$T$ inductively well-founded at
$\emptyseq$'', where $T$ {\em inductively well-founded} at $u$ is itself
defined by the following clauses:
\begin{itemize}
\item if $u \notin T$ then $T$ is inductively well-founded at $u$
\item if, for all $a$, $T$ is inductively well-founded at $u \star a$,
  then $T$ is inductively well-founded at $u$
\end{itemize}
Then, it can be proved that $T$ inductively well-founded
at~$u$ is itself not different from the existence of an intensional
tree $t$ (hidden in the structure of any proof of inductive
well-foundedness) such that $t$ realises $\lambda u'.T(u@u')$. This
justifies our claim that bar induction is at the end a way to produce
an intensionally well-founded tree from an extensionally well-founded one.


Now, if bar induction can be considered as an extensionality
principle, it should be the same for its contrapositive which is
logically equivalent to the \emph{axiom of dependent
  choice}.
This means that it should eventually be possible to rephrase the 
axiom of dependent choice as a principle asserting that, if a
tree is {\em coinductively ill-founded}, then it is {\em extensionally
ill-founded} (i.e. an infinite branch can be found). We will investigate this direction
in Section~\ref{sec:bardc}, together with precise relations
between these principles and their restriction on finitely-branching
trees, namely {\em Kőnig's Lemma}\footnote{The spelling König's Lemma is also common.
  We respect here the original Hungarian spelling of the author's name.}
and the {\em Fan Theorem}, introducing a
systematic terminology to characterise and compare these different
variants.

Note in passing that the approach to consider bar induction and choice principles as
extensional principles is consistent with the methodology developed e.g. by
Coquand and Lombardi: to avoid the necessity of choice or bar induction
axioms, mathematical theorems are restated using the (co-)inductively-defined notions of
well- and ill-foundedness rather than the extensional
notions~\cite{Coquand04,CoquandLombardi06}.

\subsection{Weak K\texorpdfstring{ő}{\"o}nig's Lemma at the intersection of Boolean Prime Filter Theorem and Dependent
  Choice}

We know from classical reverse mathematics of the subsystems of second order
arithmetic~\cite{Simpson09} that the binary form of Kőnig's
lemma, namely \hypertarget{WKL}{\emph{Weak Kőnig's Lemma}} ($\WKL$) has
the strength of {\em G\"odel's completeness theorem} (for a countable language). Classical reverse
mathematics of the axiom of choice and its variants in set
theory~\cite{Henkin49a,RubinRubin63,Jech73,Espindola16} also tells that G\"odel's completeness
theorem has the strength of the {\em Boolean Prime Filter Theorem} (for a
language of arbitrary cardinal). This suggests that the Boolean
Prime Filter Theorem is the ``natural'' generalisation of $\WKL$ from
countable to arbitrary cardinals.

On the other side, Weak Kőnig's Lemma is a
consequence\footnote{Note that Kőnig's Lemma is a theorem of set
  theory and that we need to place ourselves in a sufficiently weak
  metatheory, e.g. $\RCA_0$, to state this result.} of the axiom of
Dependent Choice, the same way as its contrapositive, the Weak Fan Theorem, is
an instance of Bar Induction, itself related to the contrapositive of the
axiom of Dependent Choice. This suggests that there is common
principle which subsumes both the Axiom of Dependent Choice and the Boolean
Prime Filter Theorem with Weak Kőnig's Lemma at their intersection.

Such a principle is stated in Section~\ref{sec:acgen} where it is
shown that the ill-founded version indeed generalises the axiom of
Dependent Choice and the well-founded version generalises Bar
Induction. In the same section, we also show that one of the instance
of the ill-founded version captures the general Axiom of Choice, but
that, in its full generality, the new principle is actually
inconsistent.

Section~\ref{sec:BPI} is devoted to show that the Boolean Prime Filter
Theorem is an instance of the generalised axiom of Dependent
Choice. In particular, this highlights that the notions of {\em ideal}
and {\em filter} generalise the notion of a binary tree where the prefix
order between paths of the tree
 is replaced by an inclusion order between
    non-sequentially-ordered paths now seen as finite
    approximations of a function from $\N$ to the two-element set $\Bool$.

\begin{table}
  \caption{Summary of logical correspondences}\label{table:summary-correspondence}
\begin{small}
  \begin{center}
    \renewcommand*{\arraystretch}{1.5}

\begin{tabular}{|p{0.5cm}|p{4cm}|p{3.2cm}|}
  \hline
{\it ref.} & \hfil {\it ill-foundedness-style} & \hfil {\it well-foundedness-style}\\
  \hhline{|-|-|-|}
& \multicolumn{2}{|c|}{\em $T$ branching from $\N$ over arbitrary $B$}\\
\hhline{|~|-|-|}
{\scriptsize Th.~\ref{thm:dc-bi}} & $\GDC_{\N BT}$ = $\DC^{\mathit{productive}}_{BT}$
                          &  $\GBI_{\N BT}$ = $\BI^{\mathit{ind}}_{BT}$ \\
{\scriptsize Th.~\ref{thm:spread-prod}} & \hspace*{1.065cm} = $\DC^{\mathit{spread}}_{BT}$
                          & \hspace*{0.90cm} = $\BI^{\mathit{barricaded}}_{BT}$ \\
{\scriptsize Th.~\ref{thm:dc}} & \hspace*{1.065cm} = $\DCSERIAL{BRb_0}$  & \\
\hhline{|~|-|-|}
& \multicolumn{2}{|c|}{\em $T$ branching from $\N$ over non-empty finite $B$}\\
\hhline{|~|-|-|}
{\scriptsize Th.~\ref{thm:gdc-kl}} & $\GDC_{\N BT}$ = $\KL^{\mathit{productive}}_{BT}$ & $\GBI_{\N BT}$ = $\FT^{\mathit{ind}}_{BT}$ \\
                & \hspace*{1.065cm} = $\KL^{\mathit{spread}}_{BT}$ & \hspace*{0.9cm} = $\FT^{\mathit{barricaded}}_{BT}$ \\
                & \hspace*{1.095cm} =${}_{\mathit{co}-\mathit{intuit.}}$ $\!\KL^{\mathit{unbounded}\!\!}_{BT}$ & \hspace*{0.93cm} =${}_{\mathit{intuit.}}$ $\!\FT^{\mathit{uniform}}_{BT}$\\
                & \hspace*{1.07cm} =${}_{\mathit{co}-\mathit{intuit.}}$ $\!\KL^{\mathit{staged}}_{BT}$
                & \hspace*{0.90cm} =${}_{\mathit{intuit.}}$ $\!\FT^{\mathit{staged}}_{BT}$ \\
  \hhline{|~|-|-|}
& \multicolumn{2}{|c|}{\em
  functions from $\N$ to arbitrary $B$}\\
  \hhline{|~|-|-|}
{\scriptsize Th.~\ref{thm:cc}} & $\GDC_{\N BR_{\top}}$ = $\CC_{BR}$ & \\
\hhline{|~|-|-|}
& \multicolumn{2}{|c|}{\em
  functions from arbitrary $A$ to arbitrary $B$}\\
  \hhline{|~|-|-|}
{\scriptsize Th.~\ref{thm:ac}} & $\GDC_{ABR_{\top}}$ = $\AC_{ABR}$ & \\
\hhline{|~|-|-|}
%
& \multicolumn{2}{|c|}{\em binary branching from arbitrary $A$}\\
\hhline{|~|-|-|}
{\scriptsize Th.~\ref{thm:compl}} & $\GDC_{A \Bool{\cal T}^C}$ = $\Compl^-_A(\scalebox{0.9}{$\cal T$})$ & $\GBI_{A \Bool{\cal T}}$ = $\Compl^+_A(\scalebox{0.9}{$\cal T$})$ \\
{\scriptsize Th.~\ref{thm:bpf}} & $\GDC_{A \Bool{\cal T}^C}$ = $\BPF_{\free{A}}(F_{\cal T})$ &\\
\hhline{|-|-|-|}
\end{tabular}
\end{center}
\end{small}

\end{table}

\subsection{Methodology and summary}

For our investigations to apply both to classical and
to intuitionistic mathematics, we carefully distinguish between the
choice axioms (seen as ill-foundedness extensionality schemes) and bar
induction schemes (seen as well-foundedness extensionality schemes).

All in all, the correspondences we obtain are
summarised in Table~\ref{table:summary-correspondence} where the definitions of the
different notions can be found in the respective sections of the
paper.

\section{The logical structure of dependent choice and bar induction principles}
\label{sec:bardc}


\subsection{Metatheory}

We place ourselves in a metatheory capable to express arithmetic
statements. In
addition to the type $\N$ of natural numbers together with induction
and recursion, we assume the following constructions to be available:
\begin{itemize}
\item The type $\Bool$ of Boolean values $0$ and $1$ together with a
  mechanism of definition by case analysis. It shall be convenient to
  allow the definition of propositions by case analysis as in
  $\myifbool{b}{P}{Q}$, whose logical meaning shall be equivalent to $(b=1
  \land P) \lor (b=0 \land Q)$.
\item For any type $A$, the type $A^*$ of finite sequences over~$A$
  whose elements shall generally be ranged over by the letters
  $u$, $v$ ... We write $\emptyseq$ for the empty sequence and
  $u \star a$ for the extension of sequence $u$ with element $a$. We
  write $|u|$ for the length of $u$ and $u(n)$ for the $\nth{n}$
  element of $u$ when $n<|u|$. We write $v @ u$ for the
  concatenation of $v$ and $u$. We write $u \leq_s v$ to mean
  that $u$ is an initial prefix of $v$. This is inductively defined by:
$$
\seq{}
    {u \leq_s u}
\qquad
\seq{u \leq_s v}
    {u \leq_s v \star a}
$$
  We shall also support case analysis over finite sequences under the
  form of a $\mathsf{case}$ operator.
\item For any two types $A$ and $B$, the type $A \arrow B$ of
  functions from $A$ to $B$. Functions can be built by
  $\lambda$-abstraction as in $\lambda x.\,t$ for $x$ in $A$ and $t$ in
  $B$ and used by application as in $t(u)$ for $t$ in $A \arrow B$ and
  $u$ in $A$. To get closer to the traditional notations, we shall
  also abbreviate $t(u_1)\ldots(u_n)$ into $t(u_1,\ldots,u_n)$.
\item A type $\Prop$ reifying the propositions as a type. The type $A
  \arrow \Prop$ shall then represent the type of predicates over~$A$.
  We shall allow predicates to be defined inductively (smallest
  fixpoint) or coinductively (greatest fixpoint), using respectively the
  $\mu$ and $\nu$ notations.
\item For any type $A$ and predicate $P$ over $A$, the subset\break
  $\{a:A~\mid~P(a)\}$ of elements of $A$ satisfying $P$.
\end{itemize}
This is a language for higher-order arithmetic but in practice, we
shall need quantification just over functions and predicates
of (apparent) rank 1 (i.e. of the form $A_1 \arrow \ldots \arrow
A_n \arrow A$ or $A_1 \arrow \ldots \arrow A_n \arrow \Prop$ with no
arrow types in $A$ and the $A_i$). We however also allow arbitrary type constants
to occur, so we can think of our effective metatheory as a
second-order arithmetic generic over arbitrary more complex types.
In practise, our metatheory could typically be the image of arithmetic
in set theory or in an impredicative type theory. We will in any case
use the notation $a \in A$ to mean that $a$ has type $A$ when $A$ is a
type, which, if in set theory, will become $a$ belongs to the set $A$.

The metatheory can be thought as classical, i.e. associated to a
classical reading of connectives but in practice, unless stated
otherwise, most statements will have proofs compatible with a linear,
intuitionistic or co-intuitionistic reading of connectives too. Using
linear logic as a reference for the semantics of
connectives~\cite{Girard87}, $A \imp B$, $\forall a\,Q$, $A \lor B$,
$\exists a\,Q$, $\neg A$ have respectively to be read linearly as $P
\multimap Q$, $\&_{a}\, Q$, $A \oplus B$, $\bigoplus_{a}\, Q$ and the
logical dual $A^{\bot}$ of $A$, while $A \land B$ has to be read $A
\otimes B$ when used as the dual of $A \imp B^\bot$ and $A \& B$ when
used as the dual of $A^\bot \lor B^\bot$. An intuitionistic reading will
add a ``!''  (of-course connective of linear logic) in front of
negative connectives while a co-intuitionistic reading will add a
``?'' (why-not connective of linear logic) in front of positive
connectives.

\subsection{Infinite sequences}
\label{sec:infinite-sequence}

We write $A^{\N}$ for the infinite (countable) sequences of elements of
$A$. There are different ways to represent such an infinite sequence:
\begin{itemize}
\item We can represent it as a function, i.e. as a functional object
  of type $\N \arrow A$.
\item We can represent it as a total functional relation, i.e. as a relation
  $R$ of type $\N \arrow A \arrow \Prop$ such that $\forall n\, \exists!
  a\, R(n,a)$.
\item Additionally, when $A$ is $\Bool$, an extra possible
  representation is as a predicate $P$ over $\N$ with intended
  meaning $1$ if $P(n)$ holds and $0$ if $\neg P(n)$ holds (and
  unknown meaning otherwise).
\end{itemize}
The representation as a functional relation is weaker in the sense
that a function $\alpha$ induces a functional relation $\lambda
n.\,\lambda a.\,\alpha(n)=a$ but the converse requires the axiom of unique
choice. In the sequel, we will use the notation $\alpha(n) \myeq a$
and $\alpha(n) \defeq a$ to mean different things
depending on the representation chosen for $\N \arrow B$.

In the first case, $\alpha(n) \myeq a$ means $\alpha(n) =_A a$
where $=_A$ is the equality
on $A$.  Similarly, $\alpha(n) \defeq a$ defines the function
$\alpha \defeq \lambda n.\,a$.

In the second case, $\alpha(n) \myeq a$ however means
$\alpha(n,a)$ and $\alpha(n) \defeq a$ defines the functional
relation $\alpha \defeq \lambda (n,a').\,(a' = a)$ where $n$ can occur in $a$.

When $A$ is $\Bool$, the representation as a predicate $P$ is even weaker
in the sense that a functional relation $R$ induces a predicate
$\lambda n.\,R(n,1)$ but the converse requires classical reasoning. We
can easily turn a predicate $P$ into a relation $\lambda n.\,\lambda
b.\,(\myifbool{b}{P(n)}{\neg P(n)})$ but proving $\forall n\,\exists!b\,
(\myifbool{b}{P(n)}{\neg P(n)})$ requires a call to excluded-middle on $P(n)$.

When $A$ is $\Bool$ and $\alpha$ is a predicate, we
define $\alpha(n) \myeq 1$ as $\alpha(n)$ and $\alpha(n) \myeq 0$ as
$\neg \alpha(n)$. Technically, this means seeing $\alpha(n) \myeq b$ as a
notation for ``$\myifbool{b}{\alpha(n)}{\neg \alpha(n)}$''.
Similarly, $\alpha(n) \defeq b$ defines $\alpha \defeq \lambda n.\,
(\myifbool{b}{\top}{\bot})$.

In particular, this means that {\em all choice and bar induction
  statements of this paper} have two readings of a different
logical strength (depending on the validity of the axiom of unique
choice in the metatheory), or even three readings (depending on the
validity of the axiom of unique choice {\em and} of classical
reasoning) when the codomain of the function mentioned in the theorems
is $\Bool$.

If $\alpha \in A^{\N}$, we write $u \prec_s \alpha$ to mean that $u$ is
an initial prefix of $\alpha$. This is defined inductively by the
following clauses:
$$
\seq{}
    {\emptyseq \prec_s \alpha}
\qquad
\seq{u \prec_s \alpha \qquad \alpha(|u|) \myeq a}
    {u \star a \prec_s \alpha}
$$

If $a \in A$ and $\alpha \in A^{\N}$, we write $a @ \alpha$
for the sequence $\beta$ defined by $\beta(0) \defeq a$ and $\beta(n+1) \defeq
\alpha(n)$.


We have the following easy property:

\begin{prop}
\label{prop:prec-star}
If $u \prec_s \alpha$ then $a @ u \prec_s a @ \alpha$.
\end{prop}


\subsection{Trees and monotone predicates}

Let $B$ be a type and $T$ be a predicate on $B^*$. We overload the notation $u \in T$ to mean that $T$
holds on $u\in B^*$. We say that $T$ is finitely-branching if $B$ is
in bijection with a non-empty bounded subset of $\N$ (i.e. to
$\{n:\N \mid n \leq p\}$ for some $p$).

We say that $T$ is a {\em tree} if it is closed under restriction,
and, dually, that $T$ is {\em monotone} if it is closed under
extension (the formal definitions are given in Table~\ref{table:2}).
Classically, we have $T$ monotone iff $\neg T$ is a tree, and,
dually, $\neg T$ monotone iff $T$ is a tree. In particular, another
way to describe a tree is as an antimonotone predicate\footnote{From a
  categorical perspective, a tree is a contravariantly functorial
  predicate over the preorder generated by $u \leq_s v$, while a
  monotone predicate is covariantly functorial.}. It is convenient for
the underlying intuition to restrict oneself to predicates which are
trees, or which are monotone, even if it does not always matter in
practice. When it matters, a predicate is turned into a tree either by
discarding sequences not connected to the root or by
completing it with missing sequences from the root: these are
respectively the {\em downwards arborification} $\innerneg{T}$ and
{\em upwards arborification} $\outerofneg{T}$ of a predicate, as shown in Table~\ref{table:2bis}. We
dually write $\outerofpos{T}$ and $\innerpos{T}$ for the {\em upwards
  monotonisation} and {\em downwards monotonisation} of $T$.
Arborification and monotonisation are idempotent. We shall in general
look for minimal definitions of the
concept involved in the paper, and thus consider arbitrary predicates
as much as possible, turning them into trees or monotone predicates
only when needed to give sense to the definitions.

\begin{table}
  \caption{Logically equivalent dual concepts on dual predicates}
  \label{table:2}
\renewcommand*{\arraystretch}{1.4}
\begin{center}
\begin{tabular}{|p{4.1cm}|p{4.1cm}|}
\hline
\hfil $T$ is a tree & \hfil $T$ is monotone \\
\hfil (closure under restriction) & \hfil (closure under extension) \\
\hfil $\forall u\,\forall a\, (u \star a \in T \imp u \in T)$ &
\hfil $\forall u\,\forall a\, (u \in T \imp u \star a \in T)$ \\
\hline
\end{tabular}
\end{center}
\end{table}

\begin{table}
  \caption{Logically opposite closure operators on dual predicates}
  \label{table:2bis}
\renewcommand*{\arraystretch}{1.4}
\begin{center}
\begin{tabular}{|p{4.1cm}|p{4.1cm}|}
\hline
\hfil {downwards arborification of $T$} &
\hfil {upwards monotonisation of $T$} \\
\hfil {($\innerneg{T}$)} &
\hfil {($\outerofpos{T}$)} \\
\hfil $\lambda u.\, \forall u'\, (u' \leq_s u \imp u' \in T)$ &
\hfil $\lambda u.\, \exists u'\, (u' \leq_s u \land u' \in T)$ \\
\hhline{-|-}
\hfil {upwards arborification of $T$} &
\hfil {downwards monotonisation of $T$} \\
\hfil {($\outerofneg{T}$)} &
\hfil {($\innerpos{T}$)} \\
\hfil $\lambda u.\, \exists u'\, (u \leq_s u' \land u' \in T)$ &
\hfil $\lambda u.\, \forall u'\, (u \leq_s u' \imp u' \in T)$ \\
\hline
\end{tabular}
\end{center}
\end{table}

\subsection{Well-foundedness and ill-foundedness properties}

We list properties on predicates which are relevant for stating
ill-foundedness axioms (i.e. choice axioms), and their dual
well-foundedness axioms (i.e. bar induction axioms). Duality can be
understood both under a classical or linear interpretation of the
connectives, where the predicate $T$ in one column is supposed to be
dual of the predicate $T$ occurring in the other column (dual
predicates if in linear logic, negated predicates if in classical
logic). Table~\ref{table:progressing} details properties which differ
by contraposition and are thus logically equivalent (in classical and
linear logic). On the other side, tables~\ref{table:foundedness} and
\ref{table:foundedness-relativised} detail properties which are
logically opposite.

\begin{table}[h]
  \caption{Basic logically equivalent dual properties on dual predicates}
\label{table:progressing}
\renewcommand*{\arraystretch}{1.4}
\begin{center}
\begin{tabular}{|p{4cm}|p{4cm}|}
\hline
\hfil $T$ is progressing at $u$ (*) &
\hfil $T$ is hereditary at $u$ \\
\hfil $u \in T \mymultimap (\exists a\, u\star a \in T)$ &
\hfil $(\forall a\, u \star a \in T) \mymultimap u \in T$ \\
\hline
\hfil $T$ is progressing (*) &
\hfil $T$ is hereditary \\
\hfil $\forall u\, (\mbox{$T$ is progressing at $u$})$ &
\hfil $\forall u\, (\mbox{$T$ is hereditary at $u$})$ \\
\hline
\end{tabular}

\end{center}
\end{table}


\begin{table}
  \caption{Logically opposite dual concepts on dual predicates}
\label{table:foundedness}
\renewcommand*{\arraystretch}{1.4}
\begin{center}
\begin{tabular}{|p{4cm}|p{4cm}|}
\hline
\hfil {\em ill-foundedness properties} & \hfil {\em well-foundedness properties} \\
\hhline{-|-}
\multicolumn{2}{|c|}{\em closure operators}\\
\hhline{-|-}
\hfil pruning of $T$ & \hfil hereditary closure of $T$ \\
\hfil $\nu X.\lambda u.\, (u \in T \mywith \exists a\, u \star a \in X)$ &
\hfil $\mu X.\lambda u.\, (u \in T \myoplus \forall a\, u \star a \in X)$ \\
\hhline{-|-}
\multicolumn{2}{|c|}{\em intensional concepts}\\
\hhline{-|-}
\hfil $T$ is a spread & \hfil $T$ is barricaded (*) \\
\hfil $\emptyseq \in T \mytensor T \mbox{ progressing}$ &
\hfil $T \mbox{ hereditary} \mymultimap \emptyseq \in T$ \\
\hline
\hfil $T$ is productive & \hfil $T$ is inductively barred \\
\hfil $\emptyseq \in \mbox{pruning of $T$}$ &
\hfil $\emptyseq \in \mbox{hereditary closure of $T$}$ \\
\hhline{-|-}
\multicolumn{2}{|c|}{\em intensional concepts relevant for the finite case}\\
\hhline{-|-}
\hfil $T$ has unbounded paths & \hfil $T$ is uniformly barred \\
\hfil $\forall n\, \exists u\, (|u|=n \mytensor u \in \innerneg{T})$ &
\hfil $\exists n\, \forall u\, (|u|=n \mymultimap u \in \outerofpos{T})$ \\
\hline
\hfil $T$ is staged infinite & \hfil $T$ is staged barred (*) \\
\hfil $\forall n\, \exists u\, (|u|=n \mytensor u \in T)$ &
\hfil $\exists n\, \forall u\, (|u|=n \mymultimap u \in T)$\\
\hhline{-|-}
\multicolumn{2}{|c|}{\em extensional concepts}\\
\hhline{-|-}
\hfil $T$ has an infinite branch & \hfil $T$ is barred \\
\hfil $\exists \alpha\, \forall u\, (u \prec_s \alpha \mymultimap u \in T)$ &
\hfil $\forall \alpha\, \exists u\, (u \prec_s \alpha \mytensor u \in T)$ \\
\hline
\end{tabular}
\end{center}
\end{table}
We indicated with (*) concepts for which we did not find an existing
terminology in the literature. Thus, the terminology is ours. Also,
what we called {\em staged infinite} is often simply called {\em
  infinite}. We used {\em staged infinite} to make explicit the
difference from a definition based on the presence of an infinite
number of nodes. Thereby we also obtain
a symmetry with the notion of {\em staged barred}. What we call {\em having an
  infinite branch} could alternatively be called {\em ill-founded}, or
{\em having a choice function}. In particular, the terminology {\em
  having an infinite branch} applies here to any predicate and is not
restricted to trees. Note that {\em well-founded} in the standard
meaning is the same as {\em barred} for the dual predicate. In
particular, when opposing ill-foundedness and well-foundedness, we
adopt a bias towards the tree view, i.e. towards the left column.
\begin{table}
  \caption{Useful relativisation of some of the concepts of Table~\ref{table:foundedness}}
  \label{table:foundedness-relativised}
\renewcommand*{\arraystretch}{1.4}
\begin{center}
\begin{tabular}{|p{4cm}|p{4cm}|}
\hline
\hfil {\em ill-foundedness-style} & \hfil {\em well-foundedness-style} \\
\hhline{-|-}
\multicolumn{2}{|c|}{\em relativised intensional concepts}\\
\hhline{-|-}
\hfil $T$ is productive from $u$ & \hfil $T$ is inductively barred from $u$ \\
\hfil $u \in \mbox{pruning of $T$}$ &
\hfil $u \in \mbox{hereditary closure of $T$}$ \\
\hhline{-|-}
\multicolumn{2}{|c|}{\em relativised intensional concepts relevant for the finite case}\\
\hhline{-|-}
\hfil $T$ has unbounded paths from $u$ & \hfil $T$ is uniformly barred from $u$\\
\hfil $\forall n\, \exists u'\, (|u'|=n \mytensor u @ u' \in \innerneg{T})$ &
\hfil $\exists n\, \forall u'\, (|u'|=n \mymultimap u @ u' \in \outerofpos{T})$ \\
\hhline{-|-}
\multicolumn{2}{|c|}{\em extensional concepts}\\ %
\hhline{-|-}
\hfil $T$ has an infinite branch from $u$ & \hfil $T$ is barred from $u$ \\
\hfil $\exists \alpha\, \forall u'\, (u' \prec_s \alpha \mymultimap u @ u' \in T)$ &
\hfil $\forall \alpha\, \exists u'\, (u' \prec_s \alpha \mytensor u @ u' \in T)$ \\
\hline
\end{tabular}
\end{center}
\end{table}
We have the following:
\begin{prop}
\label{prop:tree-unbounded}
If $T$ is a tree, then having unbounded paths is equivalent to being staged
infinite. Dually, if $T$ is monotone, being a uniform bar is
equivalent to being staged barred.
\end{prop}

\begin{proof}
Because trees and monotone predicates are invariant under arborification
and monotonisation. 
\end{proof}

As a consequence, it is common to use the notion of staged infinite,
which is simpler to formulate, when we know that $T$ is a
tree. Otherwise, if $T$ is an arbitrary predicate which is not
necessarily a tree, there is no particular interest in using
the notion of staged infinite. Similarly, staged barred is a simpler
way to state uniformly barred when $T$ is monotone, i.e., conversely,
uniform bar is the expected refinement of staged barred when $T$ is
not known to be monotone.

A {\em progressing} $T$ may be {\em productive at} $\emptyseq$ without
being productive at all $u \in T$, so we may need to {\em prune} $T$ to extract from
it a {\em spread}. Dually, not all {\em barricaded} predicates are {\em inductive bars at} all $u$
but we can saturate them into inductive bars, by taking the {\em
  hereditary closure}. We make this formal in the following
proposition:

\begin{prop}
\label{prop:productive-pruning-spread}
If $T$ is productive then its pruning is a spread. Dually, if $T$ is barricaded
then its hereditary closure is an inductive bar.
\end{prop}

\begin{proof}
That $\emptyseq$ is in the pruning of $T$ is direct from $T$
productive. That the pruning of $T$ is progressing on all $u$ is also
direct by construction of the pruning. The other part of the statement
is by duality.
\end{proof}

Conversely, by coinduction, the pruning of any progressing predicate $T$ contains $T$
and dually, induction shows that the hereditary closure of an hereditary predicate $T$
is included in $T$. Thus, we have:
\begin{prop}
$T$ spread implies $T$ productive, and, dually, $T$ inductively barred
implies $T$ barricaded. \hfill $\Box$
\end{prop}

We can then relate productive and spread, as well as inductive bar and
barricaded as follows:

\begin{prop}
\label{prop:spread-productive}
$T$ is productive iff there exists $U \subseteq T$ which is a spread. Dually,
$T$ is an inductive bar iff all $U \supseteq T$ is barricaded.
\end{prop}

\begin{proof}
By duality, it is enough to prove the first equivalence. From left to
right, we use Prop.~\ref{prop:productive-pruning-spread},
observing that the pruning of $T$ is included in $T$. From right to
left, a spread is productive and a coinduction suffices to prove that
inclusion preserves productivity.
\end{proof}

On the other side, having unbounded paths is equivalent to being a
spread or to being productive only when $T$ is finitely-branching. Similarly
for being uniformly barred compared to being an inductive bar or being barricaded.
Moreover, none of the equivalences hold linearly. The second one
requires intuitionistic logic, i.e. requires the ability to use an
hypothesis several times while the first one, dually, requires a bit
of classical reasoning\footnote{or, to be more precise,
  co-intuitionistic reasoning, that is, using a multi-conclusion
  sequent calculus to formulate the reasoning, with the contraction
  rule allowed on conclusions but not on hypotheses}.

For $S$ being a class of formulae and $P$ and $Q$ ranging over~$S$,
let $D_S$ be the principle $ \forall x y\,(P(x) \lor Q(y))) \imp
(\forall x\,P(x)) \lor (\forall y\,Q(y))$. Dually, let $C_S$ be
$(\exists x\, P(x)) \mywith (\exists y\, Q(y)) \mymultimap \exists x \exists
y\,(P(x) \mywith Q(y))$.

\begin{prop}
\label{prop:productive-unbounded}
If $B$ is non-empty finite, then productive is equivalent to having
unbounded paths and being an inductive bar is equivalent to uniformly
barred. The first statement holds in a logic where $D_S$ holds and the
second in a logic where $C_S$ holds, for $S$ a class of formulae
containing arithmetical existential quantification over $T$.
\end{prop}

\begin{proof}
Relying on duality, we only prove the first statement. Based on our 
definition of finite, we also assume without loss of
generality that $B$ is $\Bool$. Our proof relies on an
argument found in \cite{Berger09,Ishihara05} and proceeds by proving
more generally for $u \in \innerneg{T}$ that $T$ is productive from $u$ iff $T$ has unbounded
paths from $u$.

From left to right, we reason by induction on $n$. If $n$ is $0$ this
is direct from $T$ productive by defining $u' \defeq
\emptyseq$. Otherwise, by $T$ productive from $u$, we get $a$ such
that $T$ is productive from $u\star a$, obtaining by induction $u'$ of
length $n-1$ such that $(u\star a) @ u' = u @(a @u') \in
\innerneg{T}$, showing that $a@u'$ is the expected sequence of length
$n$.

From right to left, we reason coinductively. To prove that $u\in T$,
we take a path of length $0$. Then, in order to apply the coinduction
hypothesis and prove the coinductive part, we prove that there is $b$
such that $T$ has unbounded paths from $u \star b$. By $D_S$, it is
enough to prove that for all $n_0$ and $n_1$, there is a path $u_0$ of
length $n_0$ and a path $u_1$ of length $n_1$ such that either $(u
\star 0) @ u_0$ or $(u \star 1) @ u_1$ is in $\innerneg{T}$.  So, let
$n_0$ and $n_1$ be given lengths. By unbounded paths from $u$, we get
a sequence $u''$ of length $\mathit{max}(n_0,n_1)+1$ such that $u @
u'' \in \innerneg{T}$. This is a non-empty sequence, hence a sequence
of the form $b @ u'$ so that we have either $(u \star 0) @ u' \in
\innerneg{T}$ or $(u \star 1) @ u' \in \innerneg{T}$ for $u'$ of
length $\mathit{max}(n_0,n_1)$. By closure of $\innerneg{T}$, prefixes
$u_0$ of length $n_0$ and $u_1$ of length $n_1$ of $u'$ can be
extracted which both are in $\innerneg{T}$.
\end{proof}

\textit{Remark:} Based on the decomposition of $\WKL$ for decidable trees into
a choice principle and the Lesser Limited Principle of Omniscience (LLPO), we
suspect that we actually have the stronger result that the equivalence
of unbounded paths and productivity implies $D_S$ for the
corresponding underlying class of formulae $S$, and similarly with
$C_S$ and the dual statement.

\omitnow{
\subsection{The different equivalent formulations of Kőnig's Lemma and the Fan Theorem}

The previous results tend to relativise the importance of the
definitional variations in Kőnig's Lemma and the Fan Theorem.
suggests

the presence of the classical principe $D_S$ be considered
}


\subsection{Bar induction and tree-based dependent choice}

In the first part of Table~\ref{table:5}, we reformulate using our definitions
the standard statement of bar induction and a tree-based formulation of dependent choice
from the literature. The standard form of Bar Induction, as e.g. in~\cite{KleeneVesley65},
corresponds in our classification to $\BI^{\mathit{ind}}_{BT}$, apart from the fact
that we do not fix in advance the logical complexity of $B$ -- such as
being countable or not -- or the arithmetic strength of $T$ -- i.e.\
whether it is decidable, or recursively enumerable, etc. For 
dependent choice\footnote{or dependent choice{\em s} for some
  authors, e.g.~\cite{Jech73}}, we consider here a pruned-tree-based
definition $\DC^{\mathit{spread}}_{BT}$ corresponding to the instance
$\DC_{\aleph_0}$ of Levy's family of Dependent
Choice indexed on cardinals~\cite{Levy64}\footnote{Alternatively, it can be seen as the
  generalisation to arbitrary codomains of the Boolean dependent
  choice principle $DC^{\lor}$ described e.g. in
  Ishihara~\cite{Ishihara05}.}. A comparison with other logically
equivalent definitions of dependent choice will be given in
Section~\ref{sec:standard-DC}.

These formulations of Tree-based Dependent Choice and Bar Induction
are not dual\footnote{This might be related to coinductive reasoning
  historically coming later and being less common than inductive
  reasoning in mathematics.} of each other but
Prop.~\ref{prop:spread-productive} gives us a way to connect
each one with the dual of the other:

\begin{theorem}
\label{thm:spread-prod}
As schemes, generalised over $T$, $\DC^{\mathit{spread}}_{BT}$ and $\DC^{\mathit{productive}}_{BT}$ are equivalent, and so are
$\BI^{\mathit{barricaded}}_{BT}$ and $\BI^{\mathit{ind}}_{BT}$.~$\Box$
\end{theorem}

\subsection{K\texorpdfstring{ő}{\"o}nig's Lemma and the Fan Theorem}

The second part of Table~\ref{table:5} is about Kőnig's Lemma and the
Fan Theorem.

The Fan Theorem is sometimes stated over finitely-branching trees,
where the definition of finite itself may
vary~\cite{KleeneVesley65,Ishihara05}, but it is also sometimes
considered by default to be on a binary
tree~\cite{JBerger05,BergerIshihara05,Berger09,Cheung15,Coquand04,Ishihara06}
in which case the finite version is sometimes called extended. We
call here Fan Theorem the finite version, for {\em finite} defined as being in bijection with a finite prefix of $\N$, and for all branchings being on the same finite $B$. The statement of the Fan Theorem
sometimes relies on the notion of inductive bar
(e.g. \cite{Coquand04}), what we call here $\FT^{\mathit{ind}}_{BT}$, or on the definition of staged barred for
monotone predicates (as a variant in~\cite{Ishihara06}), called here $\FT^{\mathit{staged}}_{BT}$, or on the
dual notions of finite tree (i.e., technically of staged barred for the
negation of a tree) and well-founded tree (i.e., technically of
inductively barred for the negation of a tree) in
e.g.~\cite{BergerIshiharaSchuster12}, which respectively corresponds to
$\FT^{\mathit{staged}}_{BT^C}$ and $\FT^{\mathit{ind}}_{BT^C}$
for $T^C$ the complement of $T$. But it also often relies on the definition of uniform
bar~\cite{JBerger05,Berger09,BergerIshihara05,Cheung15,Ishihara05,Ishihara06,KleeneVesley65}
over an arbitrary predicate, what we call here $\FT^{\mathit{uniform}}_{BT}$. Note that, as in the case of bar
induction, we omit the usual restriction of the
statement of the Fan Theorem to decidable predicates.

Kőnig's Lemma is generally stated as $T$ infinite tree implies $T$ has an
infinite branch, but the definition of $T$ infinite may differ from
author to author. The definition in~\cite{BergerIshiharaSchuster12,Ishihara05} expresses
explicitly that the infinity can only be in depth. It does so by requiring
arbitrary long branches rather than an infinite number of nodes. The exact
definition of arbitrarily long branches also depends on authors. For
instance, \cite{Veldman14} relies (up to classical reasoning) on
having unbounded paths for arbitrary predicates rather than trees, what we call here
$\KL^{\mathit{unbounded}}_{BT}$, but most of the time it is about what
we call staged infinite
tree~\cite{Berger09,Ishihara05,Ishihara06}, leading formally to the definition $\KL^{\mathit{staged}}_{BT}$.
The versions $\KL^{\mathit{staged}}_{BT}$ and
$\KL^{\mathit{unbounded}}_{BT}$ imply LLPO~\cite{Ishihara90}.
Contrastingly, the versions which we call $\KL^{\mathit{spread}}_{BT}$
and $\KL^{\mathit{productive}}_{BT}$ are ``pure choice'' versions not
implying LLPO (see Prop.~\ref{prop:productive-unbounded} for the
connection). The binary variant $\KL^{\mathit{spread}}_{\Bool T}$ of the former occurs for instance in
the literature with name $C_{\mathsf{WKL}}$~\cite{Berger09}.


\begin{table}
\hypertarget{dict}{\caption{Tree-based dependent choice and bar induction dual principles}}
\label{table:5}
\renewcommand*{\arraystretch}{1.4}
\begin{center}
\begin{tabular}{|p{4cm}|p{4cm}|}
\hline
\hfil {\em ill-foundedness-style} & \hfil {\em well-foundedness-style} \\
\hhline{-|-}
\multicolumn{2}{|c|}{\textit{$T$ branching over arbitrary $B$}}\\
\hhline{-|-}
Tree-based Dependent Choice\hfill\break ($\DC^{\mathit{spread}}_{BT}$) &
  Alternative Bar Induction\hfill\break ($\BI^{\mathit{barricaded}}_{BT}$) \\
$T$ spread $\imp$\hfill\break \hspace*{1mm} $T$ has an infinite branch &
  $T$ barred $\imp$\hfill\break \hspace*{1mm} $T$ is barricaded\\
\hhline{-|-}
Alternative Tree-based Dependent Choice ($\DC^{\mathit{productive}}_{BT}$) &
  Bar Induction ($\BI^{\mathit{ind}}_{BT}$) \\
$T$ productive $\imp$\hfill\break \hspace*{1mm} $T$ has an infinite branch &
 $T$ barred $\imp$\hfill\break \hspace*{1mm}  $T$ inductively barred\\
\hhline{-|-}
\multicolumn{2}{|c|}{\textit{$T$ branching over non-empty finite $B$}}\\
\hhline{-|-}
$\KL^{\mathit{spread}}_{BT} \defeq \DC^{\mathit{spread}}_{BT}$ (finite $B$)  &
$\FT^{\mathit{barricaded}}_{BT} \defeq \BI^{\mathit{barric.}}_{BT}$ (fin. $B$) \\
\hhline{-|-}
$\KL^{\mathit{productive}}_{BT} \defeq \DC^{\mathit{prod.}}_{BT}$ (fin. $B$)  &
$\FT^{\mathit{ind}}_{BT} \defeq \BI^{\mathit{ind}}_{BT}$ (finite $B$) \\
\hhline{-|-}
Alternative Kőnig's Lemma ($\KL^{\mathit{unbounded}}_{BT}$) &
  Fan Theorem ($\FT^{\mathit{uniform}}_{BT}$) \\
$T$ with unbounded paths $\imp$ &
  $T$ barred $\imp$ \\[-2mm]
\hspace*{1mm} $T$ has an infinite branch &
\hspace*{1mm}   $T$ uniform bar\\
\hhline{-|-}
Kőnig's Lemma ($\KL^{\mathit{staged}}_{BT}$) &
  Staged Fan Theorem ($\FT^{\mathit{staged}}_{BT}$) \\
$T$ staged-infinite tree $\imp$ &
  $T$ barred and monotone $\imp$ \\[-2mm]
\hspace*{1mm} $T$ has an infinite branch &
\hspace*{1mm}   $T$ staged barred \\
\hhline{-|-}
\hline
\end{tabular}
\end{center}
\end{table}

\smallskip

%

There is a standard way to go from arbitrary
predicates to trees or monotone predicates by associating to each
predicate its (downward or upwards) tree or monotone closure. This allows to show that
it is equivalent to state Kőnig's Lemma on trees using
staged-infinity or on
arbitrary predicates using unbounded paths, and,
similarly, that it is
equivalent to state the Fan Theorem on monotone predicates using
staged barred ($\FT^{\mathit{staged}}_{BT}$) or on arbitrary
predicates using uniformly barred.

\begin{prop}
\label{prop:kl-ft-variants}
As schemes, when generalised over $T$, $\KL^{\mathit{staged}}_{BT}$ is
equivalent to $\KL^{\mathit{unbounded}}_{BT}$ and
$\FT^{\mathit{staged}}_{BT}$ to $\FT^{\mathit{uniform}}_{BT}$.
\end{prop}

\begin{proof}
We treat the first equivalence. From left to right, if $T$ is a
predicate, we apply $\KL^{\mathit{staged}}_{BT}$ to
$\innerneg{T}$. The resulting infinite branch is an infinite branch in
$T$ because $\innerneg{T} \subseteq T$. From right to left, the statement holds by
Prop.~\ref{prop:tree-unbounded}. The second equivalence is by
duality.
\end{proof}






\subsection{Choice and bar induction as relating intensional and extensional concepts}

The intensional definitions are stronger than the extensional ones,
which implies that the choice and bar induction axioms can
alternatively be seen as stating the logical equivalence of the
intensional and extensional versions of 
ill-foundedness and well-foundedness properties (of various
strengths). 

\begin{theorem}
\label{thm:converse}
$T$ inductively barred implies $T$ barred. Dually, 
$T$ has an infinite branch implies $T$ is productive.
\end{theorem}

\begin{proof}
We prove by induction on the definition of $T$ inductively barred that
$T$ inductively barred at $u$ implies $T$ barred from $u$ where the
latter requires that for all $\alpha$, there is $u' \prec_s \alpha$
such that $u @ u' \in T$.

If $u \in T$, then it is enough to take $\emptyseq$ for $u'$ to get $u
@ \emptyseq \in T$ for any $\alpha$. If $T$ is barred from $u \star b \in T$ for all
$b \in B$, this means that there is $u' \prec_s \beta$
such that $(u \star b) @ u' \in T$ for any $\beta$. For a given $\alpha$,
set $b \defeq \alpha(0)$ and $\beta(n) \defeq \alpha(n+1)$ so that we
can find $u' \prec_s \beta$, hence $b @ u' \prec_s b @ \beta$, i.e. $b
@ u' \prec_s \alpha$ (by Prop.~\ref{prop:prec-star}) together with $u @ (b @ u')
\in T$.

The dual proof builds $T$ productive at $u$ from $T$ has an infinite branch
from $u$ by coinduction. From the infinite branch $\alpha$ from $u$
and $\emptyseq \prec_s \alpha$ we get $u @ \emptyseq \in T$, i.e. $u \in
T$. It remains to find $b$ such that $T$ is productive from $u \star
b$ and it suffices to take $\alpha(0)$ since $T$ has an infinite
branch $\beta(n) \defeq \alpha(n+1)$ from $u \star \alpha(0)$ simply
because $v \prec_s \alpha$ implies $\alpha(0) @ v \prec_s \alpha(0) \star
\beta$ (by Prop.~\ref{prop:prec-star}) and $(u \star \alpha(0)) @ v \in T$
from $u @ (\alpha(0) @ v) \in T$.
\end{proof}

\subsection{Relation to other formulations of Dependent Choice and to countable Zorn's Lemma}
\label{sec:standard-DC}

For $R$ a relation on $B$, it is common to formulate
dependent choice as $$\begin{array}{l}\forall b^B\, \exists {b'}^B\, R(b,b') \imp\\
\qquad\forall {b_0}^B\, \exists f^{\N \arrow B}\, (f(0)=b_0 \land \forall n\,
R(f(n),f(n+1)))\,.\end{array}$$ Let us call {\em serial} a (homogeneous) relation such that
$\forall b^B\, \exists {b'}^B\, R(b,b')$ holds. In this section, we formally compare
the resulting statement of dependent choice to
$\DC_{BT}^{\mathit{productive}}$, examining also dual statements.

Let $R$ be a serial relation, i.e.\ a relation such that $\forall b^B\,
\exists {b'}^B\, R(b,b')$. Using a seed $b_0$, each such
relation $R$ can be turned into a predicate on $B^*$ under the two following ways:
\begin{itemize}
\item The {\em chaining} $R^*_{\top}(b_0)$ from $b_0$ is probably the most
  natural one: it says that $u \in R^*_{\top}(b_0)$ if all steps in
  $u$ from $b_0$ are in $R$.
\item The {\em alignment} $R^{\triangleright}_{\top}(b_0)$ from $b_0$
  artificially uses non-empty sequences to represent pairs of
  elements. We have $u \in R^{\triangleright}_{\top}(b_0)$ either when
  $u$ has at least two elements and the last two elements are related
  by $R$, or, when the sequence contains exactly one element which is
  related to $b_0$, or, finally, when the sequence is simply empty.
\end{itemize}

Reasoning by induction on $v \leq_s u$ in one direction and on $u$
in the other direction, we can show that both are related:
\begin{prop}
$u \in R^*_{\top}(b_0)$ iff $u \in
\innerneg{R^{\triangleright}_{\top}(b_0)}$
\hfill $\Box$
\end{prop}

Dually, we can define {\em antichaining} and {\em blockings} such that:
\begin{prop}
$u \in R^*_{\bot}(b_0)$ iff $u \in
\outerofpos{R^{\triangleright}_{\bot}(b_0)}$
\hfill $\Box$
\end{prop}

The formal definitions are given in Table~\ref{table:8}, where we can notice that the
use of $\mu$ vs.\ $\nu$ does not matter in practice since the structure
of the relation is a function of $|u|$.

\newcommand{\casedeux}[5]{\setlength{\arraycolsep}{0.15em}\parbox[t]{3cm}{$\mathsf{case}~#1~\mathsf{of}~\\[0.5em]
    \hspace*{0.1cm}\left[\begin{array}{lcl}#2&\mapsto&#3\\#4&\mapsto&#5\end{array}\right]\!\!\\[0.2em]$}}
\newcommand{\casetrois}[7]{\setlength{\arraycolsep}{0.15em}\parbox[t]{3cm}{$\mathsf{case}~#1~\mathsf{of}~\\[0.5em]
    \hspace*{0.1cm}\left[\begin{array}{lcl}#2&\mapsto&#3\\#4&\mapsto&#5\\#6&\mapsto&#7\end{array}\right]\\[0.2em]$}}

\begin{table}
  \caption{Logically opposite dual concepts on dual homogeneous relations}
  \label{table:8}
\renewcommand*{\arraystretch}{1.4}
\begin{center}
\begin{tabular}{|p{4.0cm}|p{4.0cm}|}
\hline
\hfil {\em ill-foundedness-style} & \hfil {\em well-foundedness-style} \\
\hhline{-|-}
\multicolumn{2}{|c|}{\em intensional concepts}\\
\hhline{-|-}
\hfil {$R$ serial} & \hfil {$R$ has a ``least'' element}\\
\hfil {$\forall b\, \exists b'\, R(b,b')$} &
\hfil {$\exists b\, \forall b'\, R(b,b')$} \\
\hhline{-|-}
\hfil {$R$ left-not-full (*)} & \hfil {$R$ has a ``maximal'' element}\\
\hfil {$\forall b\, \exists b'\, \neg R(b,b')$} &
\hfil {$\exists b\, \forall b'\, \neg R(b,b')$} \\
\hhline{-|-}
\hfil {chaining of $R$ from $b_0$ ($R^{*}_{\top}(b_0)$)} &
\hfil {antichain.\ of $R$ from $b_0$ ($R^{*}_{\bot}(b_0)$)}\\
$\mu X.\lambda b.\,\lambda u.\,$\hfill\break$\hspace*{0.1cm} \casedeux{u}{\emptyseq}{\top}{b' \star u}{R(b,b') \land X(b',u)}$ &
$\nu X.\lambda b.\,\lambda u.\,$\hfill\break$\hspace*{0.1cm} \casedeux{u}{\emptyseq}{\bot}{b' \star u}{R(b,b') \lor X(b',u)}$ \\
\hhline{-|-}
\hfil {alignment of $R$ from $b_0$ ($R^{\triangleright}_{\top}(b_0)$)} &
\hfil {blockings of $R$ from $b_0$ ($R^{\triangleright}_{\bot}(b_0)$)}\\
$\lambda u.\,\casetrois{u}{\emptyseq}{\top}{b}{R(b_0,b)}{u' \star b \star b'}{R(b,b')}$ &
$\lambda u.\,\casetrois{u}{\emptyseq}{\bot}{b}{R(b_0,b)}{u' \star b \star b'}{R(b,b')}$ \\
\hline
\end{tabular}
\end{center}
\end{table}

We are now in position to state in Table \ref{table:9} a relatively standard form of
Dependent Choice which we call $\DCSERIAL{BRb_0}$ for $R$
being a relation on $B$ and $b_0$ a seed in $B$. Though to our knowledge uncommon in the
literature, we also mention its dual which we call
$\BILEAST{BRb_0}$.

\begin{table}
\hypertarget{dictdc}{\caption{Dependent choice and bar induction principles}}
  \label{table:9}
\renewcommand*{\arraystretch}{1.4}
\begin{center}
\begin{tabular}{|p{3.8cm}|p{4.2cm}|}
\hline
\hfil {\em ill-foundedness-style} & \hfil {\em well-foundedness-style} \\
\hhline{-|-}
Dependent Choice ($\DCSERIAL{BRb_0}$) &
  Dual to Dependent Choice ($\BILEAST{BRb_0}$) \\
$R$ serial $\imp$\hfill\break \hspace*{1mm} $R^{\triangleright}_{\top}(b_0)$ has an infinite branch &
  $R^{\triangleright}_{\bot}(b_0)$ barred $\imp$\hfill\break \hspace*{1mm} $R$ has a least element\\
\hline
\end{tabular}
\end{center}
\end{table}

We state a few results that allow to show the equivalence of
$\DCSERIAL{BRb_0}$ and $\DC_{BT}^{\mathit{productive}}$
as schemes.

We have the following properties.

\begin{prop}
\label{prop:serial-productive}
$R$ serial implies $R^{\triangleright}_{\top}(b_0)$ productive for any $b_0$. Dually,
if $R^{\triangleright}_{\bot}(b_0)$ is inductively barred then $R$ has a least element.
\end{prop}

\begin{proof}
We prove by coinduction that $u \in R^{\triangleright}_{\top}(b_0)$
implies $R^{\triangleright}_{\top}(b_0)$ productive from $u$. If $u$
is empty, $R^{\triangleright}_{\top}(b_0)(\emptyseq)$ holds by
definition and there is by seriality a $b_1$ such that
$R^{\triangleright}_{\top}(b_0)(b_1)$. This allows to conclude by
coinduction hypothesis. If $u$ has the form $u' \star b$, there is
also by seriality a $b'$ such that $R^{\triangleright}_{\top}(b_0)(u'
\star b \star b')$ and we can again conclude by coinduction hypothesis.
The productivity of $R^{\triangleright}_{\top}(b_0)$ finally follows
because $R^{\triangleright}_{\top}(b_0)(\emptyseq)$ holds by
definition. The dual statement is by dual (inductive) reasoning.
\end{proof}

Conversely, for $T$ a predicate, let $B_T$ be defined by $B_T \defeq
\{u \in B^*\mid\mbox{$T$ is productive from $u$}\}$ and let $R_T$ be
the relation on $B_T$ defined by $R_T(u,u') \defeq \exists b\,(u \star
b = u')$.  The relation $R_T$ is serial by construction: for $u$ such
that $T$ is productive from $u$, there is $a$ such that $T$ is
productive from $u \star a$ and $u \star a \in T$. Also, $\emptyseq \in
B_T$ as soon as $T$ is productive.

We can now formally state the correspondence in our language:

\begin{theorem}
\label{thm:dc}
As schemes, $\DC_{BT}^{\mathit{productive}}$ and
$\DCSERIAL{BRb_0}$ are logically equivalent.
\end{theorem}

\begin{proof}
From left to right, we take $R^{\triangleright}_{\top}(b_0)$ and use
Prop.~\ref{prop:serial-productive}. From right to left, we take
$B_T$ and $R_T$, obtaining $\emptyseq \in B_T$ from $T$ productive. We
get an infinite branch $\beta$ of elements of $B_T$ such that $u
\prec_s \beta$ implies $(R_T)^{\triangleright}_{\top}(\emptyseq)(u)$,
which means first that $R_T(\emptyseq,\beta(0))$, thus $\beta(0) = b$
for some $b$, then, secondly, that for all $n$,
$R_T(\beta(n),\beta(n+1))$, i.e. $\beta(n+1)=\beta(n) \star b$ for
some $b$. It is then enough to define $\alpha(n)$ to be the
corresponding $b$ to get an infinite branch of elements of $B$. Let us
now consider $u \prec_s \alpha$. We already know $\emptyseq \in T$
from $T$ productive. Otherwise, for $u$ non empty, we get by induction
that $u$ coincides with $\beta(|u|-1)$ which is in $T$ because $u \in
B_T$ implies $T$ being productive from~$u$.
\end{proof}

As a final remark, let us mention countable Zorn's
lemma~\cite{Wolk83}: {\em If a partial order $S$ on some set has no
  countable chain, it has a maximal element}.
It corresponds to the instantiation on
$\neg S$ of the generalisation of the scheme {\em $R^*_{\bot}(b_0)$
  barred implies $R$ has a least element} over all $b_0$, using our definitions
up to classical reasoning, and dropping the partial order
requirement. This is the case
because a least element is a maximal one in the complement of a
relation and because, classically, the barring of all antichainings of
$\neg S$ is the same as the absence of countable chains in a partial order $S$.


\subsection{Relation to countable choice}
\label{sec:cc}

For $R$ heterogeneous relation on $A$ and $B$, we introduce in Table~\ref{table:10} definitions allowing to state in
Table~\ref{table:11} the axiom of countable choice, $\CC$, and its
dual, which we call {\em weak bar induction}. Note that {\em
  left-total} and {\em grounded} are respective generalisations of
serial and having a least element to non-necessarily homogeneous
relations.

\begin{table}
\caption{Logically opposite dual concepts on dual relations}
  \label{table:10}
\renewcommand*{\arraystretch}{1.4}
\begin{center}
\begin{tabular}{|p{4cm}|p{4cm}|}
\hline
\hfil {\em ill-foundedness-style} & \hfil {\em well-foundedness-style} \\
\hhline{-|-}
\hfil {$R$ $A$-$B$-left-total} & \hfil {$R$ $A$-$B$-grounded (*)}\\
\hfil {$\forall a\, \exists b\, R(a,b)$} &
\hfil {$\exists a\, \forall b\, R(a,b)$} \\
\hhline{-|-}
\hfil {$R$ has an $A$-$B$-choice function} & \hfil {$R$ is $A$-$B$-barred} \\
\hfil {$\exists \alpha\, \forall a\, \forall b\, (\alpha(a) \myeq b \imp R(a,b))$} &
\hfil {$\forall \alpha\, \exists a\, \exists b\, (\alpha(a) \myeq b \land R(a,b))$} \\
\hhline{-|-}
\end{tabular}
\end{center}
\end{table}

\begin{table}
\hypertarget{dictcc}{\caption{Countable choice and weak bar induction principles}}
  \label{table:11}
\renewcommand*{\arraystretch}{1.4}
\begin{center}
\begin{tabular}{|p{3.8cm}|p{4.2cm}|}
\hline
\hfil {\em ill-foundedness-style} & \hfil {\em well-foundedness-style} \\
\hhline{-|-}
\hfil {Countable Choice ($\CC_{BR}$)} &
  \hfil {Dual to Countable Choice ($\CBI_{BR}$)} \\
\hfil {$R$ $\N$-$B$-left-total $\imp$\hfill\break \hspace*{1mm} $R$ has an $\N$-$B$-choice function} &
\hfil {$R$ $\N$-$B$-barred $\imp$\hfill\break \hspace*{1mm} $R$ $\N$-$B$-grounded}\\
\hline
\end{tabular}
\end{center}
\end{table}

\begin{table}
\caption{Logically opposite dual concepts on dual relations}
  \label{table:12}
\renewcommand*{\arraystretch}{1.4}
\begin{center}
\begin{tabular}{|p{4cm}|p{4cm}|}
\hline
\hfil {\em ill-foundedness-style} & \hfil {\em well-foundedness-style} \\
\hhline{-|-}
\multicolumn{2}{|c|}{\em intensional concepts}\\
\hhline{-|-}
\hfil {seq.\ pos.\ alignment of $R$ ($R^{\N}_{\top}$)} &
\hfil {seq.\ neg.\ alignment of $R$ ($R^{\N}_{\bot}$)} \\
\hfil {$\lambda u.\, \casedeux{u}{\emptyseq}{\top}{u \star b}{R(|u|,b)}$} &
\hfil {$\lambda u.\,\casedeux{u}{\emptyseq}{\bot}{u \star b}{R(|u|,b)}$} \\
\hline
\end{tabular}
\end{center}
\end{table}

We shall prove that $\CC$ is derivable from $\DC^{\mathit{productive}}$
and introduce for that the {\em alignment} of a sequential relation over
$\N \times A$ as a predicate over $A^*$ (see Table~\ref{table:12}). We have:

\begin{theorem}
\label{thm:cc}
For $B$ and $R$ given ($R$ relation over $\N$ and $B$),
$\CC_{BR}$ is equivalent to
$\DC_{BR^{\N}_{\top}}^{\mathit{productive}}$.
Dually, 
$\CBI_{BR}$ is equivalent to
$\BI_{BR^{\N}_{\bot}}^{\mathit{ind}}$.
\end{theorem}

\begin{proof}
The correspondence between $R$ left-total and $R^{\N}_{\top}$
productive is obtained by coinduction from left to right and, from
right to left, by extracting the $\nth{n}$ element of the proof of
$R^{\N}_{\top}$ to get the image of $n$ by $R$. The function relating
$R$ having a choice function (as a relation) and $R^{\N}_{\top}$
having a choice function (as a predicate on $B^*$) is the same. Then,
from left to right, for non-empty $u \star b \prec_s \alpha$, we have
$\alpha(|u|) \myeq b$, thus $R(|u|,b)$ and $u \in T$. From right to
left, for $n$ and $b$ such that $\alpha(n) \myeq b$, the restriction
$\alpha_{|n+1}$ of $\alpha$ to its first $n+1$ elements is in $T$, so
that $R(|\alpha_{|n}|,b)$, i.e. $R(n,b)$.  Similarly for the dual
case.
\end{proof}

We do not conversely expect to be able in general to express
$\DC^{\mathit{productive}}$ in term of $\CC$ since countable choice
is strictly weaker than dependent choice, and similarly for
$\BI^{\mathit{ind}}$ in terms of $\CBI$. However, if $B$ is countable, it is folklore that the
statements of $\DC$ and $\CC$ become mutually
expressible by classical-reasoning-based minimisation: their common
strength as choice principle then is not greater than the axiom of unique
choice. The latter itself is a tautology if functions are
represented as functional relations. It has however the logical effect
of reifying functional relations as proper functions if functions are
represented as proper objects in a functional type.
We conjecture that the equivalence of $\BI^{\mathit{ind}}$ and $\CBI$ with
countable codomain is provable intuitionistically.

\omitnow{
\subsection{Other properties expressing well-foundedness (intuitionistic case)}

We can also consider an intuitionistically weaker, though classically
trivially equivalent, notion of bar. We say that $T$ is a {\em weak
  bar}, or that it is {\em weakly barred}, if for all $\alpha \in
A^{\N}$ there classically is $n$ such that $\alpha_{|n} \in
T$.
}

\section{Non sequential generalisation of dependent choice and bar induction}
\label{sec:acgen}

In the previous section, we considered predicates branching countably many times
over a domain $B$. In this section, we investigate how to generalise
countable sequences of branchings to branching in an arbitrary order over
a non-necessarily countable domain $A$.

When $B$ is $\Bool$, we shall obtain principles equivalent to the {\em Boolean Prime
Ideal/Filter Theorem} (ill-founded case), or to the {\em Completeness
Theorem} but we shall recover the strength of dependent choice
(ill-founded case) and bar induction (well-founded case) when $A$ is
countable, that is when $A$ is in bijection with $\N$.
In particular we will obtain the strength of the Weak Fan
Theorem (well-founded case) and Weak Kőnig's Lemma (ill-founded
case), up to classical reasoning, when $A$ is countable and $B$ is $\Bool$.

For a certain instance, we will get the strength of the full axiom of
choice. However, the new principle is limited. For instance, for
$A \defeq \Bool^\N$ and $B \defeq \N$, we end up with an inconsistent axiom.


\subsection{Finite approximations of functions}

\newcommand{\mycal}[1]{#1}

Let ${\mycal A}$ be a domain whose elements are ranged over by the
letters $a$, $a'$, ... and $B$ a codomain whose elements are ranged over by the letters $b$, $b'$, ...
Let $T$ be a predicate over $({\mycal A} \times
B)^*$ i.e. over sequences of pairs in ${\mycal A}$ and $B$, thought
as a set of possible finite approximations of a function from $A$ to $B$.
We use $v$ to range over approximations.

We order $({\mycal A} \times B)^*$ by set inclusion, which we write
$\subseteq$. We overload the notations $\!\innerneg{T}$, $\outerofneg{T}$, $\outerofpos{T}$ and $\innerpos{T}$ to now
be with respect to $\subseteq$. In particular, since $v \subseteq v'$
for any $v'$ obtained from $v$ by permutation or duplication, all closures
are stable by permutation. We write $v \sim v'$ for $v
\subseteq v'$ and $v' \subseteq v$, i.e. for the equivalence of $v$ and $v'$ as
finite sets.


Note that we do not prevent that a sequence may contain several
occurrences of the same pair $(a,b)$. However, such a sequence shall
be equivalent to a sequence without redundancies (this design choice
is somewhat arbitrary, we just found it more convenient not to
enforce the absence of redundancies).

We write $(a,b) \in v$ to mean that $(a,b)$ is one of the
elements of the sequence. For $v \in ({\mycal A} \times B)^*$, we write
$\dom{v}$ for the set of $a$ such that there is some $b$ such that
$(a,b) \in v$. For $\alpha \in {\mycal A} \arrow B$ and $v \in
({\mycal A} \times B)^*$, we define $v \prec
\alpha$ to mean $\alpha(a) \myeq b$ for all $(a,b) \in v$, or more
formally for the predicate defined by the following clauses:
$$
\seq{}
    {\emptyseq \prec \alpha}
\qquad
\seq{v \prec \alpha \qquad \alpha(a) \myeq b}
    {v \star (a,b) \prec \alpha}
$$

We think of $({\mycal A} \times B)^*$ as finite approximations of
functions from ${\mycal A}$ to $B$ and of predicates over finite
approximations as constraints generating an ideal or a filter.

In Table~\ref{table:13}, we generalise the notion of productive over
(morally) trees into a coinductive notion of {\em ${\mycal
A}$-$B$-approximable} relative to a valid finite set of approximations, and dually, we generalise
the notion of inductively barred from holding on a sequence to holding
relative to a finite set of approximations.

\begin{table}
  \caption{Logically opposite dual concepts on dual predicates}
  \label{table:13}
\renewcommand*{\arraystretch}{1.4}
\begin{center}
\begin{tabular}{|p{4cm}|p{4cm}|}
\hline
\hfil {\em ill-foundedness-style} & \hfil {\em well-foundedness-style} \\
\hhline{-|-}
\multicolumn{2}{|c|}{\em intensional concepts}\\
\hhline{-|-}
\hfil $T$ ${\mycal A}$-$B$-approximable from $v$ \hfil & $T$ inductively ${\mycal A}$-$B$-barred from $v$\\
\hfil $\nu X.\lambda v.\, \left(\!\!\begin{array}{l}v \in \innerneg{T}\, \mywith\,\\ \forall\, a \notin \dom{v}\,\\\exists b\, (v \star (a,b) \in X)\end{array}\!\!\right)$ &
\hfil $\mu X.\lambda v.\, \left(\!\!\begin{array}{l}v \in \outerofpos{T}\, \myoplus\,\\ \exists\, a \notin \dom{v}\,\\\forall b\, (v \star (a,b) \in X)\end{array}\!\!\right)$ \\
\hhline{-|-}
\hfil $T$ ${\mycal A}$-$B$-approximable &
\hfil $T$ inductively ${\mycal A}$-$B$-barred \\
\hfil $T$ ${\mycal A}$-$B$-approximable from $\emptyseq$ &
\hfil $T$ inductively ${\mycal A}$-$B$-barred from $\emptyseq$ \\
\hhline{-|-}
\multicolumn{2}{|c|}{\em extensional concepts}\\
\hhline{-|-}
\hfil $T$ has an ${\mycal A}$-$B$-choice function &
\hfil $T$ is ${\mycal A}$-$B$-barred \\
\hfil $\exists \alpha\, \forall u\, (u \prec \alpha \mymultimap u \in T)$ &
\hfil $\forall \alpha\, \exists u\, (u \prec \alpha \mytensor u \in T)$ \\
\hline
\end{tabular}
\end{center}
\end{table}

\subsection{Generalised Dependent Choice and Generalised Bar Induction}

\begin{table}
\hypertarget{dictgen}{\caption{Dual axioms on dual predicates}}
  \label{table:gdc-gbi}
\renewcommand*{\arraystretch}{1.4}
\begin{center}
\begin{tabular}{|p{4cm}|p{4cm}|}
\hline
\hfil {\em ill-foundedness-style} & \hfil {\em well-foundedness-style} \\
\hhline{-|-}
Generalised Dependent Choice ($\GDC_{{\mycal A}BT}$) & Generalised Bar Induction ($\GBI_{{\mycal A}BT}$) \\
$T$ ${\mycal A}$-$B$-approximable $\imp$\hfill\break \hspace*{1mm} $T$ has an ${\mycal A}$-$B$-choice function &
$T$ ${\mycal A}$-$B$-barred $\imp$\hfill\break \hspace*{1mm} $T$ inductively ${\mycal A}$-$B$-barred \\
\hline
\end{tabular}
\end{center}
\end{table}

We state the generalisation of dependent
choice and bar induction to non-sequential choices over a
non-necessarily countable domain in Table~\ref{table:gdc-gbi}. Called $\GDC_{ABT}$ (shortly
$\GDC_{AB}$ or $\GDC$ as schemes) and $\GBI_{ABT}$ (shortly
$\GBI_{AB}$ or $\GBI$ as schemes), they are generalisations in the
sense that they respectively capture $\DC^{\mathit{productive}}$
and $\BI^{\mathit{ind}}$ for countable~${\mycal A}$, where by countable
is meant the existence of a bijection between $A$ and $\N$.

\newcommand{\ord}[1]{\mathit{ord}({#1})}
\newcommand{\ordered}[1]{||#1||}
\newcommand{\unorderedlifte}[1]{\widehat{#1}^{\raisebox{1mm}{\scriptsize $+$}}}
\newcommand{\unorderedlifta}[1]{\widehat{#1}^{\raisebox{1mm}{\scriptsize $-$}}}

To prove it, let us assume without loss of generality that ${\mycal A}$ is $\N$ itself.  We
say that $v \in (\N \times B)^*$ is {\em sequential} whenever
either $v$ is empty or $v$ has the form $v' \star (|v'|,b)$ with $v'$
itself sequential. To each $u \in B^*$ we can associate a sequential
element $\ord{u}$ by $\ord{\emptyseq} \defeq \emptyseq$ and $\ord{u
  \star b} \defeq \ord{u} \star (|u|,b)$.

To each $T$ over $(\N \times B)^*$, we can associate $\ordered{T}$ on
$B^*$ by $u \in \ordered{T} \defeq \ord{u} \in T$.
Conversely, to each $T$ over $B^*$, we can associate
$\unorderedlifte{T}$ and $\unorderedlifta{T}$ on $(\N \times B)^*$ defined respectively by $v \in
\unorderedlifte{T} \defeq \exists u\,(v = \ord{u} \land u \in T)$ and $v \in
\unorderedlifta{T} \defeq \forall u\,(v = \ord{u} \imp u \in T)$.
We have an easy property:



\begin{prop}
\label{prop:ordered-stabilility} Let $T$ a predicate over $B^*$. If $T$ is closed under restriction,
$u \in T$ iff $u \in \ordered{\!\!\outerofneg{\unorderedlifte{T}}\,}$.
If $T$ is closed under extension,
$u \in T$ iff $u \in \ordered{\!\!\innerpos{\unorderedlifta{T}}\,}$.
\hfill $\Box$
\end{prop}


\begin{prop} For $T$ over $(\N \times B)^*$ and closed under restriction,
$T$ is $\N$-$B$-approximable iff $\ordered{T}$ is productive, and, for $T$ over $B^*$ and closed under restriction,
  $\!\outerofneg{\unorderedlifte{T}}$ is $\N$-$B$-approximable iff $T$ is productive.
  Dually, for $T$ closed under extension in both cases,
  $T$ is inductively $\N$-$B$-barred iff $\ordered{T}$ is inductively
  barred, and, $\!\innerpos{\unorderedlifta{T}}$ is inductively $\N$-$B$-barred iff
  $T$ is inductively barred.
\end{prop}

\begin{proof}
By duality and Prop.~\ref{prop:ordered-stabilility}, it is
enough to prove the first item. The proof is by coinduction in both
directions.

From left to right, we prove $T$ $\N$-$B$-approximable from
$\ord{u}$ implies $\ordered{T}$ productive from $u$. We take $|u|$ for $a$ in the
definition of $\N$-$B$-approximable from~$\ord{u}$, get some $b$ and pass
it to the definition of $\ordered{T}$ productive from~$u$.

From right to left, we prove more generally that if $\ordered{T}$ is
productive from $u$ then $T$ is $\N$-$B$-approximable from
$v$ for all $v \subseteq \ord{u}$. By definition of $u \in \ordered{T}$, we
have $\ord{u} \in T$ and thus $v \in \innerneg{T}$ by closure of
$T$. Now, take $n \not\in \dom{v}$. If $n < |u|$, we set $b$ to be
$u(n)$ and apply the coinduction hypothesis with $v$ extended with $b$,
which still satisfies $v \star b \subseteq \ord{u}$ by a combinatorial
argument. If $n \geq |u|$, we explore the proof of productivity of
$\ordered{T}$ one step further, getting some $b$ such that $u \star
b \in \ordered{T}$ and $\ordered{T}$ is productive from $u \star
b$. The property $v \subseteq \ord{u \star b}$ continues to hold and we
reason by induction on $n-|u|$ until falling into the first case.
\end{proof}

Similarly, we have:

\begin{prop}
For $T$ closed under restriction in both cases,
$T$ has an $\N$-$B$-choice function iff $\ordered{T}$ has an
  infinite branch, and, $\outerofneg{\unorderedlifte{T}}$ has a $\N$-$B$-choice
  function iff $T$ has an infinite branch. Dually, for $T$ closed under
  extension in both cases, $T$ is $\N$-$B$-barred iff
  $\ordered{T}$ is barred, and, $\!\innerpos{\unorderedlifta{T}}$ is
  $\N$-$B$-barred iff $T$ is barred.
\end{prop}

\begin{proof}
By duality and Prop.~\ref{prop:ordered-stabilility}, it is
enough to prove the first item. From left to right, if
$u \prec_s \alpha$, it is enough to consider
$\ord{u} \prec \alpha$. From right to left, if $v \prec \alpha$, we
consider $u \defeq \alpha_{|n}$, i.e. the initial prefix of length $n$
of $\alpha$, where $n$ is $|v|$. We have $u \prec_s \alpha$ thus
$u \in \ordered{T}$ and $\ord{u} \in T$. Since $v \subseteq \ord{u}$, we get
$v \in T$ by closure of $T$.
\end{proof}

Consequently, we have:

\begin{theorem}
\label{thm:dc-bi}
$\!\DC_{BT}^{\mathit{productive}}\!$ iff $\GDC_{\N BT}$ and
$\BI_{BT}^{\mathit{ind}}$ iff\footnote{Classically, or, assuming decidability or monotony of $T$.
Credits: M. Baillon.}
 $\GBI_{\N BT}$.
\end{theorem}

\begin{proof}
We mediate by the property that $\GDC_{\N BT}$ is equivalent as a
scheme to its restriction to predicates $T$ closed under restriction.
Indeed, it is enough to reason with $\innerneg{T}$ knowing that
$\innerneg{T} \subseteq T$ and that $\innerneg{T}$ is the identity on
predicates closed under restriction. The other equivalence holds by
duality
\end{proof}

Now, in combination with Prop.~\ref{prop:productive-unbounded}
and~\ref{prop:kl-ft-variants} and Th.~\ref{thm:spread-prod}, we
get:
\begin{theorem}
\label{thm:gdc-kl}
As schemes, generalised over $T$, for $B$ non-empty finite,
$\GDC_{\N BT}$ is equivalent to
$\KL^{\mathit{spread}}_{BT}$ and
$\KL^{\mathit{productive}}_{BT}$, and,
in co-intuitionistic and classical logic, equivalent also to
$\KL^{\mathit{unbounded}}_{BT}$ and $\KL^{\mathit{staged}}_{BT}$.
Dually, as schemes, $\GBI_{\N BT}$ is equivalent to
$\FT^{\mathit{barricaded}}_{BT}$ and
$\FT^{\mathit{ind}}_{BT}$, and,
in intuitionistic and classical logic, equivalent also to
$\FT^{\mathit{uniform}}_{BT}$ and $\FT^{\mathit{staged}}_{BT}$.
\hfill $\Box$
\end{theorem}

\subsection{Inconsistency of the unconstrained form of Generalised Dependent Choice
and Generalised Bar Induction}

In its full generality, the generalisation of $\GDC$ and $\GBI$ obtained
by allowing non-countable branchings over an arbitrary codomain $B$
is inconsistent: for large enough $A$ and $B$,
it may happen that some $T$ is coinductively $A$-$B$-approximable without
$T$ having a (full) $A$-$B$-choice function. Indeed, take $A \defeq \Bool^{\N}$ and
$B \defeq \N$ and filter the choice function so that it is injective.
That is, we define $u \in T$ as follows: if $u$ contains
$(f,n)$ and $(f',n)$ then $f$ and $f'$ are extensionally equal.

Then, $T$ is coinductively $\Bool^{\N}$-$\N$-approximable by successively
extending $u$ with $(f,|u|)$ for any $f$ not already in $\dom{u}$. But
there is no total choice function $\alpha$ from $\Bool^{\N}$ to $\N$,
since, by Cantor's theorem, such a function is necessarily
non-injective. Thus, taking $f$ and $f'$ distinct such that
$n \defeq \alpha(f) = \alpha(f')$, we get that the sequence
$(f,n),(f',n) \prec \alpha$ is not in $T$.

Therefore, we have:
\begin{prop}
As schemes, $\GDC_{\Bool^{\N}\N T}$ and $\GBI_{\Bool^{\N}\N T}$ are inconsistent (this requires classical logic; credits: Y. Forster).
\end{prop}

\begin{table}
  \caption{Logically opposite dual concepts on dual homogeneous relations}
  \label{table:23}
\renewcommand*{\arraystretch}{1.4}
\begin{center}
\begin{tabular}{|p{4cm}|p{4cm}|}
\hline
\hfil {\em ill-foundedness style} & \hfil {\em well-foundedness-style} \\
\hhline{-|-}
\multicolumn{2}{|c|}{\em intensional concepts}\\
\hhline{-|-}
\hfil positive alignment of $R$ ($R_{\top}$) &
\hfil negative alignment of $R$ ($R_{\bot}$) \\
\hfil $\lambda v.\, \forall (a,b) \in v\,(R(a,b))$ &
\hfil $\lambda v.\, \exists (a,b) \in v\,(R(a,b))$ \\
\hline
\end{tabular}
\end{center}
\end{table}

\subsection{Relation to the  general axiom of choice}

\begin{table}
\hypertarget{dict2}{\caption{The axiom of Choice and its dual}}
  \label{table:22}
\renewcommand*{\arraystretch}{1.4}
\begin{center}
\begin{tabular}{|p{4cm}|p{4cm}|}
\hline
\hfil {\em ill-foundedness-style} & \hfil {\em well-foundedness-style} \\
\hhline{-|-}
Standard Axiom of Choice\hfill\break ($\AC_{ABR}$) &
  Dual to Standard Axiom of Choice ($\coAC_{ABR}$) \\
$R$ $A$-$B$-left-total $\imp$ &
$R$ $A$-$B$-barred $\imp$\\
\hspace*{1mm} $R$ has an $A$-$B$-choice function &
\hspace*{1mm} $R$ $A$-$B$-ground\\
\hline
\end{tabular}
\end{center}
\end{table}

We state the standard axiom of choice in Table~\ref{table:22} and prove that it is equivalent
to an instance of the generalised dependent choice $\GDC$. To do so,
we generalise in Table~\ref{table:23} the notion of sequential
alignment introduced in Section~\ref{sec:cc} to the notion of
(non-sequential) {\em alignment of a relation} on $A \times B$ as a predicate over $(A \times
B)^*$.

\begin{theorem}
\label{thm:ac}
$\AC_{ABR}$ is logically equivalent to $\GDC_{ABR_{\top}}$
\end{theorem}

\begin{proof}
The proof is a variant of the one of Th.~\ref{thm:cc}. For instance, the
correspondence between $R$ $A$-$B$-left-total and $R_{\top}$
$A$-$B$-approximable is by coinduction from left to right,
calling left-totality at each step, and, from right to left, for any~$a$,
by using $A$-$B$-approximability from $\emptyseq$ to
get $b$ such that $R(a,b)$.
\end{proof}

\section{The Boolean instances of generalised dependent choice and bar induction: relation to
the Boolean Prime Ideal/Filter Theorem and completeness theorems}
\label{sec:BPI}

\subsection{Generalised Weak K\texorpdfstring{ő}{\"o}nig Lemma and Generalised Weak Fan Theorem}

By instantiating the codomain $B$ to $\Bool$ in $\GDC_{ABT}$ and $\GBI_{ABT}$, we
obtain extensions $\GBI_{{\mycal A}\Bool T}$ of the Weak Fan Theorem
(precisely of $\FT^{\mathit{ind}}_{\Bool T}$, i.e. $\GBI^{\mathit{ind}}_{\N\Bool T}$ by
Th.~\ref{thm:gdc-kl}) and $\GDC_{{\mycal
A}\Bool T}$ of the Weak Kőnig Lemma (precisely of
$\KL^{\mathit{productive}}_{\Bool T}$, i.e. $\GDC^{\mathit{productive}}_{\N\Bool T}$
by Th.~\ref{thm:gdc-kl}) which replace the countable
sequence of branching made on a ``tree'' (in practise predicates) by a
countable sequence of choices in arbitrary order over a
non-necessarily countable domain.
This will be proved
equivalent to a version of the Boolean Prime Ideal/Filter Theorem
where primality is formulated positively and to versions of the
completeness theorem for entailment relations.
This is consistent with
the standard reverse mathematics results which show that the completeness
theorem is equivalent to the Weak Kőnig's Lemma on countable
theories~\cite{Simpson09} but equivalent to the Boolean Prime Filter
Theorem on theories of arbitrary
cardinality~\cite{Henkin49a,RubinRubin63,Jech73,Espindola16}.

\begin{table}
  \caption{Dual axioms on dual predicates}
  \label{table:14}
\renewcommand*{\arraystretch}{1.4}
\begin{center}
\begin{tabular}{|p{4cm}|p{4cm}|}
\hline
\hfil {\em ill-foundedness-style} & \hfil {\em well-foundedness-style} \\
\hhline{-|-}
Generalised Weak Kőnig's Lemma ($\GDC_{{\mycal A}\Bool T}$) & Generalised Weak Fan Theorem ($\GBI_{{\mycal A}\Bool T}$) \\
$T$ ${\mycal A}$-$\Bool$-approximable $\imp$\hfill\break \hspace*{1mm} $T$ has an ${\mycal A}$-$\Bool$-choice function &
$T$ ${\mycal A}$-$\Bool$-barred $\imp$\hfill\break \hspace*{1mm} $T$ inductively ${\mycal A}$-$\Bool$-barred \\
\hline
\end{tabular}
\end{center}
\end{table}

\subsection{Logical reading: relation to completeness theorem}
\label{sec:logical-reading}

We can give a logical reading to $({\mycal A} \times \Bool)^*$ as
follows.  We call {\em atom} any element of ${\mycal A}$. We interpret
pairs in ${\mycal A} \times \Bool$ as {\em literals}, i.e. as atoms
together with a polarity indicating whether the atom is positive or negative (we
adopt the convention that $1$ stands for positive and $0$ for
negative). We call {\em clause} any unordered sequence of elements in
${\mycal A} \times \Bool$. We call {\em context} any unordered sequence
of elements of ${\mycal A}$. We range over clauses by the letters $C$,
$D$ and over contexts by the letters $\Gamma$, $\Delta$, ...

Any clause $C$ can canonically be represented as a pair of two
contexts $\Gamma$ and $\Delta$ with $\Gamma$ the subset of positive
elements of ${\mycal A}$ in $C$ and $\Delta$ the subset of negative
elements. We write $\Gamma \triangleright \Delta$ for such a pair. We
call a set of clauses a {\em theory} and use the letter ${\cal T}$ to
range over theories. We write $(\Gamma \triangleright \Delta) \in
{\cal T}$ to mean that there is a clause of ${\cal T}$ associated to
the pair $\Gamma \triangleright \Delta$.  We write $\Gamma \intersect
\Delta$ to mean that $\Gamma$ and $\Delta$ have an atom in common.

We consider (a variant of) Scott's notion of entailment
relation~\cite{Scott74}, i.e. of a preorder relation up to ``side
contexts''. Let ${\cal T}$ be a theory on ${\mycal A}$. We define the {\em
  entailment relation} generated by ${\cal T}$ to be the smallest relation on
sequents, written $\Gamma \vdash_{\cal T} \Delta$, with $\Gamma$ and $\Delta$
treated as sets, such that the following holds:
\begin{small}
$$
\seqr{\raisebox{-0.8mm}{$\!\!\!\!\Ax$}}{\Gamma \intersect \Delta}
     {\Gamma \vdash_{\cal T} \Delta}
\quad
\seqr{\raisebox{-0.8mm}{$\!\!\!\!\Ax_{\cal T}$}}{(\Gamma \triangleright \Delta) \in \outerofpos{{\cal T}}}
     {\Gamma \vdash_{\cal T} \Delta}
\quad
\seqr{\raisebox{-0.8mm}{$\!\!\!\!\Cut$}}{\Gamma \vdash_{\cal T} \Delta, F \quad \Gamma, F \vdash_{\cal T} \Delta}
     {\Gamma \vdash_{\cal T} \Delta}
$$
\end{small}
It is usual to add an explicit weakening rule to the definition of
entailment relation but here we shall consider it as an admissible
rule.
%
%
Formally, the existence of a derivation of $\Gamma \vdash_{\cal T} \Delta$
using the inferences rules above is the same as
\begin{small}
$$\vdash_{\cal T} \defeq \mu X.\lambda (\Gamma
  \mathbin{\triangleright}
\Delta).\,\left(\!\!\begin{array}{l}(\Gamma \mathbin{\triangleright} \Delta) \in
  \outerofpos{{\cal T}}\\ \myoplus\, \exists\, F \notin (\Gamma \cup
  \Delta)\left(\!\!\begin{array}{l}(\Gamma, F \mathbin{\triangleright} \Delta) \in
                X \\ \mywith\, (\Gamma \mathbin{\triangleright} \Delta, F) \in
                X\end{array}\!\!\right)\end{array}\!\!\right)$$
\end{small}
Thus, $\Gamma \vdash_{\cal T} \Delta$ exactly says that ${\cal T}$ is
inductively ${\mycal A}$-$\Bool$-barred from $\Gamma
\mathbin{\triangleright} \Delta$.

Conversely, let us consider $\Gamma \not\vdash_{\cal T} \Delta$. We
could define it by negation of $\Gamma \vdash_{\cal T} \Delta$ but we
instead give a direct explicit definition which we call {\em positive
disprovability} and which is equivalent to the negation of $\Gamma \vdash_{\cal
  T} \Delta$ when the connectives are read linearly or classically
(though not equivalent when read intuitionistically). Let ${\cal T}^C$
denote the complement of ${\cal T}$, i.e. $(\Gamma
\mathbin{\triangleright} \Delta) \in {\cal T}^C \defeq \neg ((\Gamma
\mathbin{\triangleright} \Delta) \in {\cal T})$. The positive disprovability $\Gamma
\not\vdash_{\cal C} \Delta$ can be characterised as the
${\cal T}^C$ ${\mycal A}$-$\Bool$-approximability from $\Gamma
\triangleright \Delta$, that is, formally:
\begin{small}
$$(\Gamma \triangleright \Delta) \in
\nu X.\lambda (\Gamma \mathbin{\triangleright} \Delta).\,
\left(\!\!\begin{array}{l}(\Gamma \mathbin{\triangleright} \Delta) \in
        \innerneg{\cal T}^C\\
        \mywith\, \forall\, F \notin (\Gamma \cup \Delta)
        \left(\!\!\begin{array}{l}(\Gamma, F \mathbin{\triangleright} \Delta) \in X
                \\
                \myoplus\, (\Gamma \mathbin{\triangleright} \Delta, F) \in X
              \end{array}\!\!\right)
      \end{array}\!\!\right)
$$
\end{small}

Let $\alpha$ be a function from ${\mycal A}$ to $\Bool$. It can be interpreted
as a model over ${\mycal A}$ with $1$ to indicate that the atom is true in the
model and $0$ to indicate that the atom is false in the model.

Truth $\alpha \vDash {\cal T}$ of a theory ${\cal T}$ in a model
$\alpha$ can be defined by $$\alpha \vDash {\cal T} ~~\defeq~~ \forall
(\Gamma \mathbin{\triangleright} \Delta) \in {\cal T}\, (\Gamma \subset \alpha
\imp \Delta \intersect \alpha)$$ where we use the notation $\Gamma
\subset \alpha$ to mean that $\forall a \in \Gamma\, \alpha(a) \myeq 1$ and the notation
$\Delta \intersect \alpha$ to mean $\neg \forall a \in \Delta\,
\alpha(a) \myeq 0$. Then, ${\cal T}$ is satisfiable (or has a model) if there exists
$\alpha$ such that $\alpha \vDash {\cal T}$.

Like for disprovability, the negation of truth can be defined
explicitly rather than by negation in a way which is equivalent when
the connectives are read linearly or classically (but not
intuitionistically). Let us define {\em positive
  falsity} of a theory ${\cal T}$ in a model $\alpha$, written $\alpha
\not\vDash {\cal T}$, by the following formula:
$$\alpha \not\vDash {\cal T} ~~\defeq~~ \exists (\Gamma \mathbin{\triangleright}
\Delta) \in {\cal T}\, (\Gamma \subset \alpha \land \Delta \subset \overline{\alpha})$$
where $\Delta \subset \overline{\alpha}$ stands for
$\forall a \in \Delta\, \alpha(a) \myeq 0$. We say that the theory ${\cal T}$ is {\em positively unsatisfiable}
if, for all $\alpha$, $\alpha \not\vDash {\cal T}$.

Then, still identifying clauses in ${\cal T}$ as sequences in
$({\mycal A} \times \Bool)^*$, we get that ${\cal T}$ ${\mycal
  A}$-$\Bool$-barred corresponds to the positive unsatisfiability of ${\cal
  T}$. Also, noticing that $\exists \alpha\,
\forall u\, (u \prec \alpha \mymultimap u \in {\cal T}^{C})$ is isomorphic
to $\exists \alpha\, \forall u\, (u \in {\cal T}
\mymultimap \neg u \prec \alpha)$ and that $\neg u \prec \alpha$ is
isomorphic to $\Gamma \subset \alpha \imp \Delta \intersect
\alpha$, we get that
${\cal T}^C$ has an ${\mycal A}$-$\Bool$-choice function if and only if
there exists a model for ${\cal T}$ (see Table~\ref{table:15}
where $\vdash_{\cal T}$ and $\not\vdash_{\cal T}$ refer to the
provability and positive disprovability of the empty clause).

\begin{table}
  \caption{Logically opposite dual concepts of logic on the same predicate}
  \label{table:15}
\renewcommand*{\arraystretch}{1.4}
\begin{center}
\begin{tabular}{|p{4cm}|p{4cm}|}
\hline
\hfil {\em ill-foundedness-style} & \hfil {\em well-foundedness-style} \\
\hhline{-|-}
\multicolumn{2}{|c|}{\em intensional concepts}\\
\hhline{-|-}
\hfil ${\cal T}$ is (positively) consistent &
\hfil \quad\qquad ${\cal T}$ is inconsistent \quad\qquad\\
\hfil $\not\vdash_{\cal T}$ &
\hfil $\vdash_{\cal T}$ \\
\hhline{-|-}
\multicolumn{2}{|c|}{\em extensional concepts}\\
\hhline{-|-}
\hfil ${\cal T}$ is satisfiable & \hfil ${\cal T}$ is (positively) unsatisfiable \\
\hfil $\exists \alpha\, \alpha \vDash {\cal T}$ &
\hfil $\forall \alpha\, \alpha \not\vDash {\cal T}$ \\
\hline
\end{tabular}
\end{center}
\end{table}

\begin{table}
\hypertarget{dictcompl}{\caption{Reformulation of Table~\ref{table:14} as statements about a given logical theory}}
  \label{table:16}
\renewcommand*{\arraystretch}{1.4}
\begin{center}
\begin{tabular}{|p{4.05cm}|p{3.95cm}|}
\hline
\hfil {\em ill-foundedness-style} & \hfil {\em well-foundedness-style} \\
\hhline{-|-}
Model-existence-style Completeness\hspace*{-1cm}\hfill\break Theorem ($\Compl^-_A({\cal T})$) &
Provability-style Completeness \hfill\break Theorem ($\Compl^+_A({\cal T})$) \\
${\cal T}$ consistent $\imp {\cal T}$ is satisfiable &
${\cal T}$ unsatisfiable $\imp {\cal T}$ inconsistent \\
\hline
\end{tabular}
\end{center}
\end{table}

The completeness theorem of logic is conventionally expressed either
as the existence of a model for any consistent theory, or
contrapositively, that if a theory is unsatisfied in all theories,
then it is inconsistent, as shown on Table~\ref{table:16}.
For instance, see Rinaldi, Schuster and
Wessel~\cite{RinaldiSchusterWessel17} for the statement of a
completeness theorem such as $\mbox{Compl}^+({\cal T})$, up to the
identification of some $\exists$ with $\neg\neg\exists$. See also
e.g.~\cite{RinaldiSchuster16} for an algebraic reading.
Summing up, we have:
\begin{theorem}
\label{thm:compl}
Let ${\cal T}$ be a theory of clauses over some set of atoms $A$, with clauses represented as sequences
in $({\mycal A} \times \Bool)^*$.
The {\em Generalised Weak Kőnig's Lemma} over the complement ${\cal T}^C$ of ${\cal T}$, i.e. $\GDC_{{\mycal A}\Bool{\cal T}^C}$, coincides
with the model-existence formulation of completeness for the Scott entailment relation generated by ${\cal T}$, i.e. $\Compl^-_A({\cal T})$. Contrapositively,
the {\em Generalised Weak Fan Theorem} over ${\cal T}$, i.e. $\GBI_{{\mycal A}\Bool{\cal T}}$, coincides
with the provability-style formulation of completeness for the Scott entailment relation generated by ${\cal T}$, i.e. $\Compl^+_A({\cal T})$. Record that, to preserve the duality,
$\Compl^-_A({\cal T})$ relies on an explicit definition of $\Gamma
\not\vdash_{\cal T} \Delta$ which is linearly (and classically) equivalent to but
intuitionistically stronger than the negation of $\Gamma \vdash_{\cal
  T} \Delta$, and
$\Compl^+_A({\cal T})$ relies on an explicit definition of $\alpha \not\vDash {\cal T}$
which is linearly (and classically) equivalent to but
intuitionistically stronger than the negation of $\alpha \vDash {\cal T}$.
\hfill $\Box$
\end{theorem}

Note incidentally that entailment relations are connective-free. The usual
reliance on Markov's principle to intuitionistically prove
completeness as validity implies provability~\cite{Kreisel58} does not
apply (see e.g.~\cite{HerbelinIlik16,ForsterKirstWehr21} for recent studies).


\subsection{Algebraic reading: relation to the Boolean Prime Ideal/Filter Theorem}
\label{sec:ultrafilter-compl}

\newcommand{\clor}{\stackrel{.}{\lor}}
\newcommand{\cland}{\stackrel{.}{\land}}
\newcommand{\cbot}{\stackrel{.}{\bot}}
\newcommand{\ctop}{\stackrel{.}{\top}}
\newcommand{\cneg}{\stackrel{.}{\neg}}
\newcommand{\cvdash}{\stackrel{.}{\vdash}}
\newcommand{\cbigwedge}{\stackrel{.}{\bigwedge}}
\newcommand{\cbigvee}{\stackrel{.}{\bigvee}}
\newcommand{\clos}[1]{\mathit{clos}{#1}}

The previous reasoning based on entailment relations can also be expressed in
terms of Boolean algebras, connecting Generalised Weak Kőnig's Lemma
to the Boolean Prime Ideal/Filter Theorem. There is however a caveat:
the standard definition of proper filter and proper ideal is by negation and it
will be equivalent to approximability only with a linear or
classical, i.e. involutive, reading of the negation.

Let $({\cal B},\clor,\cland,\cbot,\ctop,\cneg)$ be a Boolean algebra and
$\cvdash$ the canonical order relation associated to it: $b \cvdash b'
\defeq (b \cland b') = b$. We call {\em filter} over ${\cal B}$ any
non-empty subset $F$ of ${\cal B}$ which is closed under $\cland$ and
closed under $\cvdash$ on the right. A filter is {\em proper} if it
does not contain $\cbot$. Otherwise, it coincides with ${\cal B}$ and
we call it {\em full}. We call {\em ultrafilter} a maximal proper
filter. A maximal filter in a Boolean algebra can be described as a
map $U$ from ${\cal B}$ to $\Bool$ such that $b_1 \cland b_2 \in U$
iff $b_1 \in U \land b_2 \in U$, $b_1 \clor b_2 \in U$ iff $b_1 \in U
\lor b_2 \in U$, $\cneg b \in U$ iff $\neg (b \in U)$, $\ctop \in U$,
and $\cbot \not\in U$. In a Boolean algebra, the notion of maximal
filter coincides with the notion of prime filter where a filter $F$ is
{\em prime} if $(b_1 \clor b_2) \in F$ implies $b_1 \in F$ or $b_2 \in
F$.

Dually, we call {\em ideal} over ${\cal B}$ any non-empty subset $I$
of ${\cal B}$ which is closed under $\clor$ and closed under $\cvdash$
on the left. An ideal is {\em proper} if it does not contain $\ctop$,
and {\em full} otherwise. A {\em prime} ideal $I$ is such that $(b_1
\cland b_2) \in I$ implies $b_1 \in I$ or $b_2 \in I$ and this
coincides with the notion of maximal proper ideal. A prime/maximal
proper ideal can be characterised in a dual way to prime/maximal
proper filter, i.e. as a map $U$ from ${\cal B}$ to $\Bool$ such that
$b_1 \cland b_2 \in U$ iff $b_1 \in U \lor b_2 \in U$, $b_1 \clor b_2
\in U$ iff $b_1 \in U \land b_2 \in U$, $\cneg b \in U$ iff $\neg (b
\in U)$, $\cbot \in U$ and $\ctop \not\in U$.

There is a family of provably equivalent theorems about the existence
of maximal/prime ideals/filters in Boolean algebras (see
e.g. Jech~\cite[2.3]{Jech73}) called Boolean Prime Ideal Theorem in
arbitrary Boolean algebras, or Ultrafilter Theorem in the Boolean
algebra of subsets of a set. We consider in Table~\ref{table:17} the case of a general
Boolean algebra and state the Boolean Prime Ideal Theorem in its two
``ideal'' and ``filter'' flavours. We also consider their
contrapositives.

\begin{table}
\hypertarget{axalg}{\caption{Reformulation of Table~\ref{table:14} as statements about a given Boolean algebra}}
  \label{table:17}
\renewcommand*{\arraystretch}{1.4}
\begin{center}
\begin{tabular}{|p{3.95cm}|p{4.05cm}|}
\hline
\hfil {\em ill-foundedness-style} & \hfil {\em well-foundedness-style} \\
\hhline{-|-}
Boolean Prime Filter Theorem ($\BPF_{{\cal B}}(F)$ for $F$ a filter) &
``Boolean Full Filter Theorem'' ($\coBPF_{{\cal B}}(F)$ for $F$ a filter) \\
$F$ proper $\imp$\hfill\break \hspace*{1mm} $F$ extensible into prime filter &
$F$ not extensible into prime filter $\imp\!\!\!$ \hfill\break \hspace*{1mm} $F$ full \\
\hhline{-|-}
Boolean Prime Ideal Theorem ($\BPI_{{\cal B}}(I)$ for $I$ an ideal) &
``Boolean Full Ideal Theorem'' ($\coBPI_{{\cal B}}(I)$ for $I$ an ideal) \\
$I$ proper $\imp$\hfill\break \hspace*{1mm} $I$ extensible into prime ideal &
$I$ not extensible into prime ideal $\imp\!\!$ \hfill\break \hspace*{1mm} $I$ full \\
\hline
\end{tabular}
\end{center}
\end{table}

We now compare the Boolean Prime Ideal/Filter Theorems to
Generalised Weak Kőnig's Lemma, i.e. $\GDC_{{\mycal A}\Bool T}$, showing first
that the Generalised Weak Kőnig's Lemma is an instance of
the Boolean Prime Ideal and Boolean Prime Filter Theorems.

To any domain ${\mycal A}$ we can associate a freely generated Boolean algebra
$(\free{{\mycal A}},\clor,\cland,\cbot,\ctop,\cneg)$ by considering the set of
algebraic expressions built from $\clor$, $\cland$, $\cbot$, $\ctop$
and $\cneg$, all quotiented by the axioms of a Boolean algebra.

As in the previous section, any $v$ in $({\mycal A} \times \Bool)^*$,
can be written under the form $\Gamma \mathbin{\triangleright} \Delta$
and a predicate over $({\mycal A} \times \Bool)^*$ can be seen as a
theory ${\cal T}$ of clauses. Let $\vdash_{{\cal T}}$ be the
associated entailment relation and $F_{\cal T}$ be the (equivalence
classes of) Boolean expressions of the form $\cbigwedge_i ((\cbigvee
\cneg\! \Gamma_i) \;\clor\; (\cbigvee \Delta_i))$ such that $\Gamma_i
\vdash_{\cal T} \Delta_i$ holds for all $i$ (this can be shown
independent of the exact choice of conjunctive normal form). It is
relatively standard to show that $F_{\cal T}$ is a filter.
This filter is
proper if $\bot \not\in F_{\cal T}$, that is if $\neg (\vdash_{\cal T})$, that is if ${\cal
T}$ is not inconsistent, that is, by
Section~\ref{sec:logical-reading}, if ${\cal T}^C$ is
$A$-$\Bool$-approximable, where the connectives are interpreted either
linearly or classically.

We can dually define $I_{\cal T}$ to be the (equivalence
classes of) Boolean expressions of the form $\cbigvee_i ((\cbigwedge
\Gamma_i) \;\cland\; (\cbigwedge \cneg\!\!  \Delta_i))$ such that
$\Gamma_i \vdash_{\cal T} \Delta_i$ holds for all $i$.
This is an ideal which is proper if $\top \not\in I_{\cal
  T}$, that is if $\neg (\vdash_{\cal T})$, that is if
${\cal T}^C$ is $A$-$\Bool$-approximable where, again, the
connectives are interpreted either linearly or classically.

Reasoning by induction on the definition of $\vdash_{\cal T}$ and relying on the definition of $(\Gamma \triangleright \Delta) \prec \alpha$,
we have the general result that prime filters and prime
ideals on a free Boolean algebra, here $\free{{\mycal A}}$, are
characterised by their intersection with generators, here~${\mycal A}$.
Whether other elements of $\free{{\mycal A}}$ belong
or not to a prime filter or prime ideal is canonically
determined\footnote{We define the value of $\alpha$ as equations to
  remain agnostic on the representation of a function to $\Bool$, see
  \ref{sec:infinite-sequence}.} by:
\begin{small}
$$
\!
\begin{array}{llll}
\alpha(a \clor a') \myeq b'' & \!\defeq\! & (\alpha(a) \myeq b) \land (\alpha(a')\myeq b') \land (b'' = b + b')\!\\
\alpha(a \cland a') \myeq b'' & \!\defeq\! & (\alpha(a) \myeq b) \land (\alpha(a')\myeq b') \land (b'' = b \cdot b')\\
\alpha(\cbot) \myeq b & \!\defeq\! & b = 0\\
\alpha(\ctop) \myeq b & \!\defeq\! & b = 1\\
\alpha(\cneg a) \myeq b' & \!\defeq\! & (\alpha(a) \myeq b) \land (b' = 1 - b)\\
\end{array}
$$
\end{small}
where $+$, $\cdot$, $-$ are the corresponding operations on $\Bool$, and where the prime filter case is characterised by $\alpha(b) \myeq 1$ and the prime ideal case by $\alpha(b) \myeq 0$.

In particular, the existence of a function from $A$ to $\Bool$
characterising a prime filter that extends the filter $F_{\cal T}$ on
$\free{\mycal A}$ is the same, by Section~\ref{sec:logical-reading},
as a model of ${\cal T}$ and as an $A$-$\Bool$-choice function for
${\cal T}^C$. By focusing on $\alpha(b) \myeq 0$ rather than
$\alpha(b) \myeq 1$, this very same function also characterises the
prime ideal that extends the ideal~$I_{\cal T}$, so, we get:

\begin{theorem}
\label{thm:bpf}
$\GDC_{{\mycal A}\Bool{\cal T}^C}$, where the connectives are
  interpreted linearly or classically, is equivalent to
  $\BPF_{\free{{\mycal A}}}(F_{\cal T})$ and $\BPI_{\free{{\mycal
        A}}}(I_{\cal T})$. \hfill $\Box$
\end{theorem}

Conversely, if $F$ is a filter on a Boolean algebra ${\cal B}$, we can
define ${\cal T}_F$ on $({\cal B} \times \Bool)^*$ by $(\Gamma
\triangleright \Delta) \in {\cal T}_F \defeq (\cbigvee \cneg\! \Gamma)
\;\clor\; (\cbigvee \Delta) \in F$.  By induction on a proof of
$\Gamma \vdash_{\cal T} \Delta$ we can show that it implies $(\cbigvee
\cneg\! \Gamma) \;\clor\; (\cbigvee \Delta) \in F$ thus $\Gamma
\vdash_{\cal T} \Delta$ iff $(\cbigvee \cneg\! \Gamma) \;\clor\;
(\cbigvee \Delta) \in F$. Therefore, $F$ proper becomes equivalent to
${\cal T}_F$ $A$-$\Bool$-approximable where the connectives are
interpreted either linearly or classically. Reasoning as above, this
eventually allow to reduce $\BPF_{{\cal B}}(F)$ to $\GDC_{{\cal B}\Bool
  {\cal T}_F}$ and to show the equivalence of $\GDC_{A\Bool T}$ and
$\BPF_{{\cal B}}(F)$ as schemes.
Then, a similar analysis can put $\GBI_{{\mycal A}\Bool{\cal T}}$ into
correspondence with $\coBPF_{\cal B}(F)$ and $\coBPI_{\cal
  B}(I)$.

More generally, we also believe that, like in the countable case,
$\GDC_{ABT}$ and $\GDC_{ABT}$ over any finite, non-necessarily
two-element, codomain $B$ can be reduced to
$\GDC_{A\Bool T}$ and $\GDC_{A\Bool T}$.

\omitnow{
Conversely, let $({\cal B},\clor,\cland,\cbot,\ctop,\cneg)$ be
a Boolean algebra. To any filter $F$ we can associate a subset $T_F$ of
$({\cal B} \times \Bool)^*$ made of the singletons $(b,0)$ for all $b
\in F$, plus, for each instance $b = b'$ of a Boolean algebra axiom,
the four two-element sequences $(b,0) \star (b',1)$,
$(b',1) \star (b,0)$, $(b,0) \star (b',1)$ and $(b',1) \star (b,0)$.

We can then characterise the property of being a proper filter as
$T_F$ being ${\cal B}$-$\Bool$-approximable.

\begin{prop}
$\cbot \in F$ iff $\vdash_{T_F}$ is derivable.
\end{prop}

In particular, $T^C$ ${\cal B}$-$\Bool$-productive expresses that the
filter generated by $T^C$ does not contain the empty sequent,
i.e. does not contain $\bot$. Similarly, $T^C$ has a ${\cal
  B}$-$\Bool$-choice function is linearly equivalent to the existence
of a maximal filter containing $T^C$...

\begin{theorem}
$\BPF_{{\cal B}}(F)$ is equivalent to $\GDC_{{\cal B}\Bool(T_F)^C}$ and
$\coBPF_{{\cal B}}(F)$ is equivalent to $\GBI_{{\cal B}\Bool T_F}$.
\end{theorem}
}



\section{Further questions}
\label{sec:questions}

The duality revealed that when a proof requires classical reasoning
and its dual does not, it is that it requires co-intuitionistic
reasoning and its dual intuitionistic reasoning. As a conclusion, to
the notable exception of Proposition~\ref{prop:productive-unbounded},
we believe that all proofs could be carried out in a linear variant of
higher-order arithmetic.

There is a rich literature on choice axioms and on principles
equivalent to choice axioms. Not all of them can be classified as
either ill- or barred/well-foundedness-style, though.
For instance, open induction and update
induction~\cite{Raoult88,Coquand97,Berger04}, are classically equivalent to
bar induction and dependent choice but are formulated as
well-foundedness of some order on functions.
The study could also for instance be extended to choice principles such as
Zorn's lemma, the ordinal variants of the axiom of dependent choices
by L\'evy~\cite{Levy64} and the ordinal variants of Zorn's
lemma~\cite{Wolk83} by Wolk.

\section*{Acknowledgments}

We thank the communities of researchers who contributed to develop the
material we built on, and in particular Camille Noûs, from the
Cogitamus Lab, who embodies the collective and collaborative nature of
scientific research.

The ideas in Section~\ref{sec:ultrafilter-compl} derived from
investigations led by Charlotte Barot~\cite{Barot17}. The second
author thanks Valentin Blot and Étienne Miquey for numerous fruitful
discussions on the axiom of dependent choice and bar
induction. Special thanks also to the reviewers for their
corrections and insightful suggestions.


\bibliographystyle{plain}
\bibliography{../biblio-en,../extra}

\nocite{Berger09}





\end{document}